\documentclass[preprint,12pt,authoryear]{elsarticle}
\usepackage{xcolor}
\usepackage{physics}
\usepackage{microtype}
\usepackage{subcaption}
\usepackage{booktabs}
\usepackage{amssymb}
\usepackage{makecell}
\usepackage{mathtools}
\usepackage{amsthm}
\usepackage{algorithm}
\usepackage{algorithmic}
\usepackage{float}
\usepackage{hyperref}
\usepackage[capitalize,noabbrev]{cleveref}
\graphicspath{{figures/}}

\theoremstyle{plain}
\newtheorem{theorem}{Theorem}[section]
\newtheorem{proposition}[theorem]{Proposition}

\newtheorem{corollary}[theorem]{Corollary}
\theoremstyle{definition}
\newtheorem{definition}[theorem]{Definition}
\newtheorem{assumption}[theorem]{Assumption}
\theoremstyle{remark}

\journal{Expert Systems with Applications}

\begin{document}
\begin{frontmatter}

\title{Quantum Fidelity Landscape-Guided Prior Calibration for Single-Circuit QGAN Image Generation}

\author{
Xue Yang\textsuperscript{a,b,c,d},
Rigui Zhou\textsuperscript{a,b,*},
Dax Enshan Koh\textsuperscript{c,e,*},
Siong Thye Goh\textsuperscript{d,f,*},
Yitao Tang\textsuperscript{g},
ShiZheng Jia\textsuperscript{a,b},
Young-Wook Cho\textsuperscript{c},
% Xuezhi Ma\textsuperscript{c},
Hongyu Chen\textsuperscript{h}
}

\address{
\textsuperscript{a}School of Information Engineering, Shanghai Maritime University, Shanghai 201306, China\\
\textsuperscript{b}Research Center of Intelligent Information Processing and Quantum Intelligent Computing, Shanghai 201306, China\\
\textsuperscript{c}Quantum Innovation Centre (Q.InC), Agency for Science, Technology and Research (A*STAR), 2 Fusionopolis Way, Innovis \#08-03, Singapore 138634, Republic of Singapore\\
\textsuperscript{d}Institute of Advanced Intelligence and Computing (IAIC), Agency for Science, Technology and Research (A*STAR), 1 Fusionopolis Way, \#16-16 Connexis, Singapore 138632, Republic of Singapore\\
\textsuperscript{e}Engineering Cluster, Singapore Institute of Technology, 1 Punggol Coast Road, Singapore 828608, Republic of Singapore\\
\textsuperscript{f}Lee Kong Chian School of Business, Singapore Management University (SMU), Singapore 178899, Republic of Singapore\\
\textsuperscript{g}Fu Foundation School of Engineering and Applied Science, Columbia University, New York, NY 10027, USA\\
\textsuperscript{h}School of Computer Science and Technology, Tongji University, Shanghai 201804, China
}

% \cortext[cor1]{Corresponding author}
% % \cortext[cor2]{Second corresponding authors}
% \emailauthor{rgzhou@shmtu.edu.cn}{Rigui Zhou}
% \emailauthor{dax.koh@singaporetech.edu.sg}{Dax Enshan Koh}
% \emailauthor{goh_siong_thye@a-star.edu.sg}{Siong Thye Goh}

% \fntext[fn1]{These authors contributed equally to this work}

\begin{abstract}
Quantum Generative Adversarial Networks (QGANs) have emerged as representative generative models in the Noisy Intermediate-Scale Quantum (NISQ) era and have attracted increasing attention in quantum machine learning. However, most existing QGAN methods rely on patch-based decomposition strategies, which weaken the global consistency of generated images and increase quantum resource overhead. In this work, we investigate a simpler approach: pixel-level, end-to-end image generation using a single-quantum-circuit QGAN. By analyzing the structural matching relationship between the quantum prior and the target data distribution in Hilbert space, we provide a new theoretical perspective for understanding the training behavior of naive end-to-end QGANs. Specifically, we introduce the \textit{Quantum Fidelity Landscape (QFL)}, defined as the pairwise-fidelity structure induced by an ensemble of quantum states and preserved under shared unitary transformations of the quantum generation process. We show that, under a fixed Lipschitz readout, this invariant imposes a one-sided bound on decoded sample separation, motivating calibration of the prior-induced QFL before adversarial training. To validate this theoretical insight, we propose BasicQGAN, a QGAN framework incorporating quantum prior calibration. Before adversarial optimization, BasicQGAN aligns the prior-induced QFL with the data-induced QFL. Experimental results on small-scale grayscale image datasets show that BasicQGAN achieves stable and effective end-to-end pixel-level image generation while requiring fewer qubits and trainable parameters than representative patch-based quantum generators. Furthermore, experiments with different initial quantum-state ensembles show that QFL-calibrated ensembles achieve better generative performance.
\end{abstract}

% QGAN  QML  GAN 
% quantum prior  image generation quantum fidelity 

\begin{keyword}
QGAN \sep QML \sep GAN \sep
quantum prior \sep image generation \sep quantum fidelity 
\end{keyword}

\end{frontmatter}

% \caption{
% \textbf{Quantum fidelity landscape (QFL) as an intrinsic constraint for end-to-end QGANs.}
% (a) An initial-state ensemble induces a pairwise-fidelity structure (QFL) that is invariant under the shared unitary generator updates.
% (b) Since measurement and subsequent classical post-processing cannot increase sample distinguishability, the data-induced target QFL imposes a one-sided (upper-bound) constraint on the quantum prior: overly high-fidelity (overly concentrated) priors limit output diversity and make matching the target distribution difficult, motivating offline QFL calibration.
% }
% \label{fig:qfl_motivation}

\section{Introduction}

% Since the real images are embedded by amplitude encoding, their pixel intensities correspond to normalized real-valued basis amplitudes. To keep the generated samples in the same real-valued image representation, we use a fixed real-part amplitude decoder. Given a generated state
% \[
% |\psi\rangle=\sum_k c_k |k\rangle,\quad c_k\in\mathbb{C},
% \]
% we compute
% \[
% \hat{x}_k =
% 2\left(
% \frac{|\operatorname{Re}(c_k)|}{\max_j |\operatorname{Re}(c_j)|+\epsilon}
% \right)-1,
% \]
% and reshape $\{\hat{x}_k\}$ into an image. This decoder is fixed throughout all experiments. It is not intended to be the unique or hardware-native readout; rather, it defines the image representation under which we evaluate QFL-based prior calibration. Probability readout based on $|c_k|^2$ is a hardware-native alternative, but it induces a different image geometry and is left for future work.

Quantum machine learning (QML) \citep{biamonte2017quantum} is an emerging interdisciplinary field at the intersection of quantum computing and classical machine learning. It leverages quantum mechanical properties to enhance classical machine learning algorithms.
Among various QML models \citep{peruzzo2014variational,farhi2014quantumapproximateoptimizationalgorithm,lloyd2018quantum}, QGANs \citep{lloyd2018quantum,dallaire2018quantum,niu2022entangling} have emerged as a promising paradigm for learning complex data distributions, particularly in image generation \citep{vieloszynski2024latentqganhybridqganclassical,thomas2025vaeqwganaddressingmodecollapse}.

\paragraph{Problem}
Quantum generative adversarial networks (QGANs) have been widely studied for modeling complex data distributions and have gradually been extended to high-dimensional tasks such as image generation \citep{stein2021qugan}. Most existing QGAN image models follow a patch-based paradigm \citep{huang2021experimental, tsang2023hybrid}, where a full image is decomposed into local patches, each patch is generated by a separate quantum subcircuit, and the patches are stitched back to form the final output. This strategy partially mitigates the dimensional burden of high-resolution images, but it also introduces two major limitations: \textit{(i)} independently generated patches often weaken global consistency and structural coherence; and \textit{(ii)} as resolution increases, both the number of quantum circuits and the parameter budget grow rapidly, leading to prohibitive quantum resource costs. \textit{We therefore focus on} end-to-end, pixel-level image generation with a \emph{single} parameterized quantum circuit, avoiding patch-wise generation and stitching altogether. The central question is why naive single-circuit QGANs often fail to train, and how to make this cleaner generation paradigm stable.

\begin{figure}[H]
    \centerline{\includegraphics[width=0.92\linewidth]{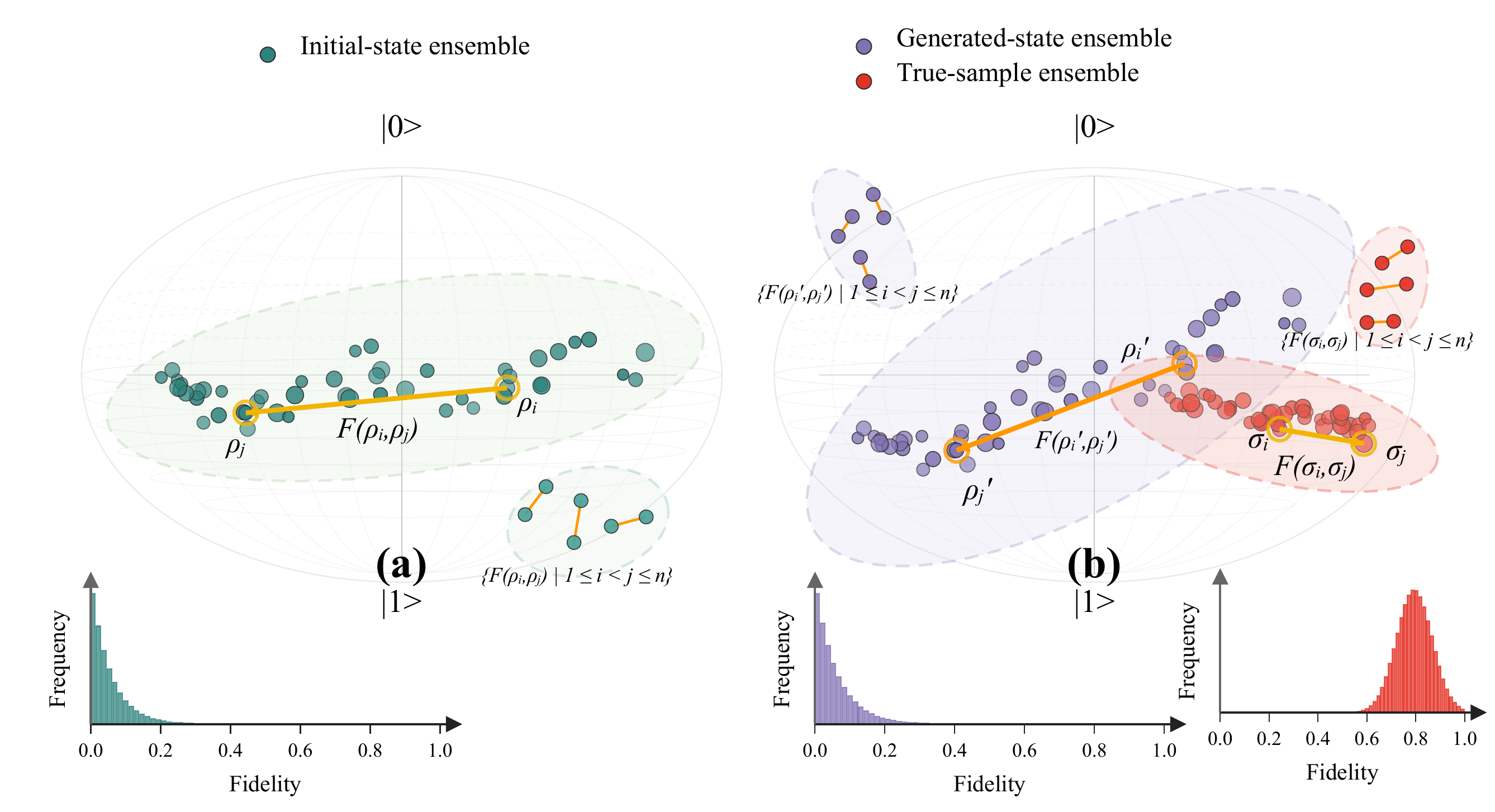}}
    \caption{The green dots in (a) denote a discrete initial-state ensemble sampled from the quantum prior, and the blue dots in (b) represent the corresponding generated-state ensemble after the generator's unitary evolution. In (b), we also show a denser ensemble example (red dots), whose QFL histogram exhibits a markedly different distribution from the other two ensembles. In each QFL histogram, the horizontal axis denotes the pairwise fidelity between distinct quantum states, while the vertical axis gives the frequency of state pairs falling into each fidelity bin; each unordered pair is counted once, with self-pairs excluded.
}
    \label{fig:qfl_motivation}
   % \vspace{-5mm}
\end{figure}

\paragraph{Diagnosis}
We identify the quantum prior as a structural factor in single-circuit QGAN training.
% Because the generator applies the same unitary to every input state, it preserves all pairwise fidelities induced by the initial-state ensemble. We call this invariant pairwise structure the \emph{Quantum Fidelity Landscape} (QFL). 
Because the generator applies the same unitary to every input state, it preserves all pairwise fidelities induced by the initial-state ensemble. Consequently, the relative positions among the states, as reflected by their pairwise fidelities, remain unchanged throughout the shared unitary evolution. We refer to this invariant pairwise structure as the \textit{Quantum Fidelity Landscape} (QFL). As illustrated in Figure ~\ref{fig:qfl_motivation}, the initial-state ensemble and its unitary-evolved generated-state ensemble therefore share the same QFL histogram, whereas an ensemble with a different state-space distribution can induce a markedly different QFL histogram.
For a fixed Lipschitz readout, we further show that QFL imposes a one-sided bound on decoded sample separation: pairs with high quantum fidelity cannot be mapped arbitrarily far apart in the output space. The amplitude-readout setting used by BasicQGAN admits a more explicit Euclidean-distance bound. These results connect the geometry of the quantum prior to the attainable geometry of generated samples and motivate calibrating the prior before adversarial training.

\paragraph{Method}
Building on this analysis, we introduce BasicQGAN, an offline--online framework for single-circuit, end-to-end image generation. In the offline stage, real images are amplitude encoded to construct the data-induced QFL, and the parameters of the initial-state preparation distribution are optimized to align the empirical prior-induced and data-induced QFL distributions. The calibration objective combines a KDE-based density discrepancy with explicit mean and variance matching. In the online stage, the calibrated initial-state ensemble is passed through a shared quantum generator, decoded by a fixed amplitude readout, and optimized adversarially with a WGAN-GP critic.

\paragraph{Evidence}
We evaluate each part of this mechanism through controlled comparisons. Using the same generator circuit, BasicQGAN produces clearer samples and lower FID than a naive global QGAN with a conventional prior. The prior ablation further shows that mismatched QFLs lead to training failure or mode collapse, whereas the calibrated ensemble produces diverse samples and achieves the lowest FID among the evaluated prior constructions. At $16\times16$ resolution, BasicQGAN is competitive with classical WGAN-GP and consistently outperforms the patch-based PQWGAN across MNIST digit classes; it also obtains lower FID than PQWGAN on the Uppercase Letters and Geometric Shapes datasets. Resource comparisons show that the single-circuit generator requires fewer qubits and trainable parameters than representative patch-based quantum generators, while the ancilla-qubit experiment demonstrates that the QFL calibration principle also applies to an auxiliary-system architecture.

Our main contributions are summarized as follows:
\begin{itemize}
  \item We formalize the Quantum Fidelity Landscape as the pairwise-fidelity structure of a quantum-state ensemble and show that it is invariant under the shared unitary evolution of a quantum generator.
  \item We connect this Hilbert-space invariant to classical generation by deriving a general one-sided output-separation bound for fixed Lipschitz readouts and an explicit Euclidean bound for amplitude readout.
  \item We develop BasicQGAN, an offline--online framework that calibrates the prior-induced QFL against the data-induced QFL before single-circuit WGAN-GP adversarial training.
  
  \item Simulations show that BasicQGAN achieves effective, generalizable generation with compact quantum resources across datasets.
  \item We further validate the effectiveness of QFL calibration through prior ablations and robustness to auxiliary-system extensions.  
  % We validate the role of QFL calibration through a controlled global baseline and prior ablations, and demonstrate competitive $16\times16$ generation, cross-dataset effectiveness, compact quantum-resource requirements, and applicability to an auxiliary-system architecture.
\end{itemize}
\section{Background and Related Work}
\label{sec:background}

\subsection{Quantum-Information Preliminaries}
We use the squared Uhlmann fidelity to quantify the similarity between two quantum
states. For density operators $\rho$ and $\sigma$, it is defined as
\begin{equation}
F(\rho,\sigma)
=
\left(
\operatorname{Tr}\sqrt{\sqrt{\rho}\,\sigma\,\sqrt{\rho}}
\right)^2.
\label{eq:prelim_fidelity}
\end{equation}
For pure states $\rho=\lvert\psi\rangle\langle\psi\rvert$ and
$\sigma=\lvert\phi\rangle\langle\phi\rvert$, this reduces to
\begin{equation}
F(\rho,\sigma)=\left|\langle\psi\mid\phi\rangle\right|^2.
\label{eq:prelim_pure_fidelity}
\end{equation}

Two standard properties of fidelity underpin our analysis
\citep{nielsen2010quantum}. First, fidelity is invariant under a common unitary
transformation:
\begin{equation}
F(U\rho U^\dagger,U\sigma U^\dagger)=F(\rho,\sigma).
\label{eq:prelim_unitary_invariance}
\end{equation}
Second, quantum measurement cannot increase state distinguishability. Let
$\mathcal{M}$ be a fixed measurement, and let $p_\rho$ and $p_\sigma$ denote
its outcome distributions for $\rho$ and $\sigma$, respectively. Then
\begin{equation}
\operatorname{TV}(p_\rho,p_\sigma)
\leq D_{\mathrm{tr}}(\rho,\sigma)
\leq \sqrt{1-F(\rho,\sigma)},
\label{eq:prelim_measurement_bound}
\end{equation}
where $D_{\mathrm{tr}}(\rho,\sigma)=\frac12\lVert\rho-\sigma\rVert_1$ is
the trace distance and
$\operatorname{TV}(p,q)=\frac12\sum_k|p_k-q_k|$ is the classical total
variation distance. The first inequality follows from data processing under
quantum measurement, and the second is the upper Fuchs--van de Graaf inequality
under the squared-fidelity convention
\citep{fuchs1999cryptographic,nielsen2010quantum}.

These properties play distinct roles in our method. Unitary invariance determines
which pairwise relations cannot be changed by the shared quantum generator, while
Eq.~\eqref{eq:prelim_measurement_bound} connects those relations to the
distinguishability of measurement outcomes.

\subsection{Related Work}
\paragraph{Classical adversarial training}
Generative adversarial networks learn a data distribution through competition between
a generator and a discriminator \citep{goodfellow2014generative}. Wasserstein GAN
replaces the original divergence-based objective with the Wasserstein-1 distance
\citep{arjovsky2017wassersteingan}, and WGAN-GP enforces the critic's Lipschitz
constraint using a gradient penalty \citep{gulrajani2017improvedtrainingwassersteingans}.
BasicQGAN adopts WGAN-GP as the objective for its online adversarial training stage.

\paragraph{Quantum adversarial learning and fidelity measures}
Foundational formulations of QGANs extended adversarial learning to quantum data and
parameterized quantum circuits \citep{lloyd2018quantum,dallaire2018quantum}.
\citet{lloyd2018quantum} further showed that quantum adversarial learning may offer an
exponential advantage when the data consist of measurement samples from certain
high-dimensional quantum systems. Fidelity has also appeared directly in QGAN
objectives: QuGAN uses swap-test estimates of quantum-state fidelity to construct
quantum generator and discriminator losses \citep{stein2021qugan}. In contrast, QFL
does not serve as the online adversarial loss between individual real and generated
states. It describes the pairwise-fidelity geometry of the entire initial-state
ensemble and the structure preserved by a shared unitary generator.

\paragraph{Quantum image-generation architectures}
Early experimental QGAN image generation includes the patch-based construction of
\citet{huang2021experimental}, which assigns image regions to multiple quantum
subcircuits and stitches their outputs. \citet{tsang2023hybrid} combined this strategy
with Wasserstein training to generate $28\times28$ grayscale images. Patch-based
generators reduce the qubit requirement of each subcircuit, but they do not explicitly
model dependencies across independently generated patches, and their total circuit and
parameter requirements grow with the number of patches. These methods also use fixed
uniform latent priors, including $U(0,\pi/2)$ for rotation angles in
\citet{huang2021experimental} and $U(0,1)$ in \citet{tsang2023hybrid}.
Other approaches generate lower-dimensional representations for subsequent classical
reconstruction. MosaiQ uses PCA-derived feature spaces \citep{silver2023mosaiq}, while
LaSt-QGAN operates in the latent space of a classical autoencoder
\citep{chang2024latentstylebasedquantumgan}. These designs improve scalability but
differ from direct pixel generation because a classical transformation or decoder maps
the quantum output back to the image space.

\paragraph{Latent sampling and processing in QGANs}
% Several recent methods modify the latent input or the quantum circuit to improve mode
% coverage and full-image generation. VAE-QWGAN learns a data-informed latent prior using
% a classical variational autoencoder and a Gaussian mixture model
% \citep{thomas2025vaeqwganaddressingmodecollapse}. ReQGAN uses a neural noise encoder
% and a trainable intensity-calibration module for single-circuit image synthesis
% \citep{yang2026end}, whereas \citet{jager2026scaling} emphasize circuit-induced
% inductive biases and enhanced noise-input techniques for end-to-end high-resolution
% generation. BasicQGAN instead treats the pairwise fidelity geometry of the quantum
% initial-state ensemble as a shared-unitary invariant and calibrates its empirical QFL
% against that of amplitude-encoded data before adversarial training. Its generative
% pathway consists of a single quantum circuit followed by a fixed amplitude readout,
% without a trainable classical decoder.
% ---- BEGIN INLINED METHOD SECTION ----
Existing image-oriented QGANs typically sample latent variables from simple prescribed
distributions. For example, PQWGAN considers latent variables drawn from either a uniform
distribution $U[0,1)$ or a standard Gaussian distribution $\mathcal{N}(0,1)$, and adopts the
uniform prior in its main experiments (Tsang et al., 2023). Similarly, LaSt-QGAN samples each
component of its latent noise independently from a standard normal distribution,
$\mathbf{z}\sim\mathcal{N}(0,I)$ (Chang et al., 2024). More recently, VAE-QWGAN replaces such
prescribed priors with a data-informed latent distribution: latent representations are first
obtained using a classical VAE encoder, and a Gaussian mixture model is subsequently fitted to
these representations for latent sampling at inference time (Thomas et al., 2025). These approaches primarily modify the distribution from which latent variables are sampled. 
% In contrast, our work focuses on the pairwise-fidelity geometry induced after latent variables are encoded into quantum states.  
Related studies have also explored alternative mechanisms for processing latent noise in quantum image generation, including the neural noise encoding used in ReQGAN \citep{yang2026endtoendqganbasedimagesynthesis} and the enhanced noise-input strategies considered by J\"ager et al. \citep{J_ger_2026}. These approaches primarily modify the distribution from which latent variables are sampled or the way latent noise is processed. In contrast, our work focuses on the pairwise-fidelity geometry induced after latent variables are encoded into quantum states.

% \textit{Quantum priors and end-to-end generation.}
% Existing image-oriented QGANs typically sample latent variables from simple prescribed distributions. For example, PQWGAN considers latent variables drawn from either a uniform distribution $U[0,1]$ or a standard Gaussian distribution $\mathcal{N}(0,1)$, and adopts the uniform prior in its main experiments \citep{tsang2023hybrid}. Similarly, LaSt-QGAN samples each component of its latent noise independently from a standard normal distribution, $z\sim\mathcal{N}(0,I)$ \citep{chang2024last}. More recently, VAE-QWGAN replaces such prescribed priors with a data-informed latent distribution: latent representations are first obtained using a classical VAE encoder, and a Gaussian mixture model is subsequently fitted to these representations for latent sampling at inference time \citep{thomas2025vae}. Related studies have also explored alternative mechanisms for processing latent noise in quantum image generation, including the neural noise encoding used in ReQGAN \citep{yang2026end} and the enhanced noise-input strategies considered by J\"ager et al. \citep{jager2026}. These approaches primarily modify the distribution from which latent variables are sampled or the way latent noise is processed. In contrast, our work focuses on the pairwise-fidelity geometry induced after latent variables are encoded into quantum states.

\section{Method: BasicQGAN}
\label{sec:method}

\subsection{Problem Setting and Overview}
Let $z$ denote a latent variable used to prepare an initial quantum state $\rho_z$.
The quantum generator applies the same parameterized unitary $U_\theta$ to every input,
and a fixed readout map $\mathcal{D}$ converts the resulting state into a classical image:
\begin{equation}
    G_\theta(z)=\mathcal{D}\!\left(U_\theta\rho_zU_\theta^\dagger\right).
    \label{eq:generator_pipeline}
\end{equation}
Unlike patch-based generators, BasicQGAN uses one shared circuit to produce the complete
pixel vector. This design is resource efficient, but it raises a structural question:
which relations among input samples can be changed by training the shared unitary?

Our analysis and method follow three steps. First, we identify a pairwise geometry of the
prior ensemble that is invariant under a shared unitary. Second, we relate this invariant
to the separation of classically decoded outputs, first for a general stable readout and
then more explicitly for amplitude readout. Third, we use the resulting principle to
calibrate the quantum prior before standard adversarial training. Figure~\ref{model}
summarizes this offline--online pipeline.

\begin{figure*}[t]
    \centering
    \includegraphics[width=1.02\linewidth]{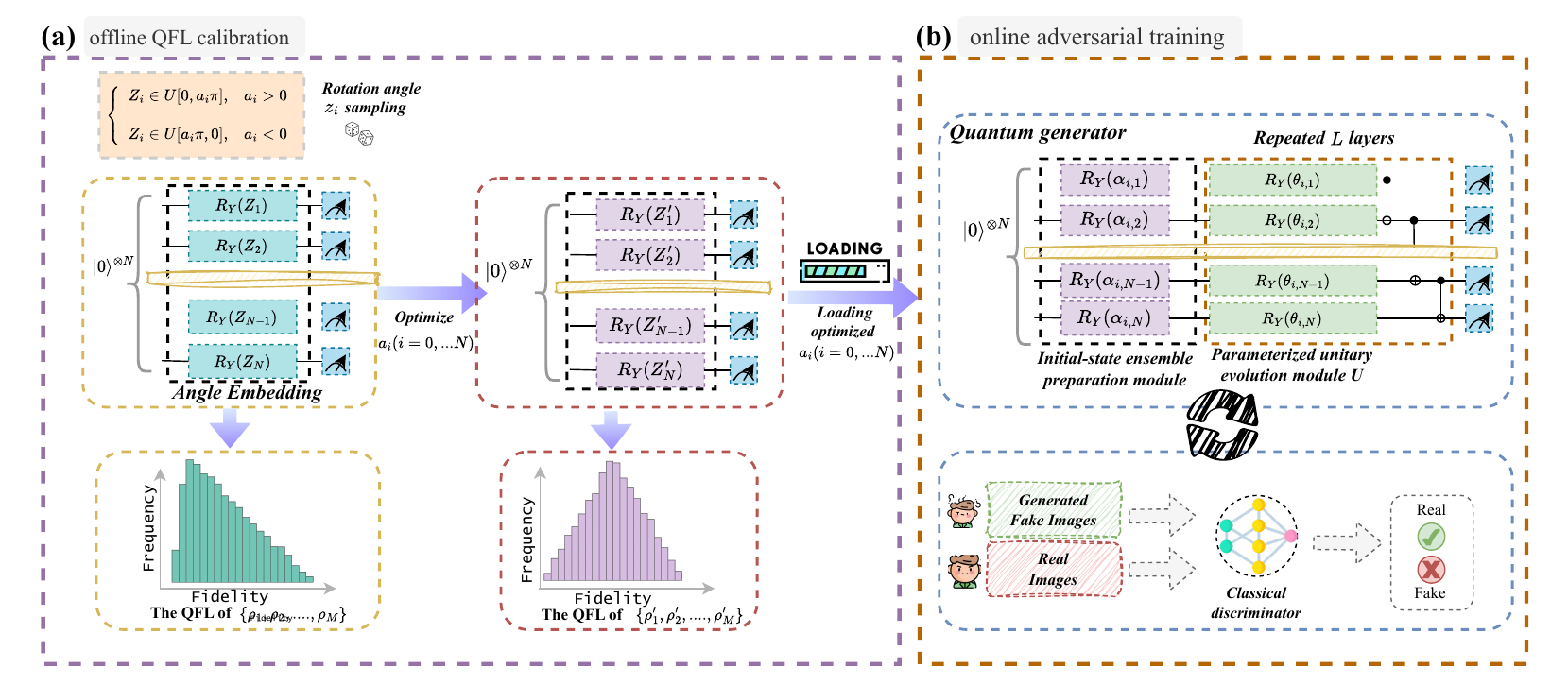}
    \caption{Overall framework of BasicQGAN. (a) During offline QFL calibration,
    the state-preparation parameters $z_1,z_2,\ldots,z_n$ are optimized to
    transform the initial ensemble $\rho_1,\rho_2,\ldots,\rho_n$ into a
    calibrated ensemble $\rho_1',\rho_2',\ldots,\rho_n'$. (b) During online
    adversarial training, the calibrated preparation parameters
    $z_1',z_2',\ldots,z_n'$ initialize the state ensemble used by the
    single-circuit quantum generator.}
    \label{model}
\end{figure*}

\subsection{QFL as a Prior-Determined Invariant}

\begin{definition}[Quantum fidelity landscape]
\label{def:qfl}
For an ensemble $\mathcal{E}=\{\rho_i\}_{i=1}^{M}$, its quantum fidelity landscape
(QFL) is the pairwise-fidelity matrix
\begin{equation}
    K_{ij}(\mathcal{E})=F(\rho_i,\rho_j),
    \qquad 1\leq i,j\leq M,
    \label{eq:qfl_matrix}
\end{equation}
where $F$ is the squared Uhlmann fidelity. For pure states,
$F(\rho_i,\rho_j)=|\langle\psi_i|\psi_j\rangle|^2$.
We use the empirical distribution of the off-diagonal entries of $K$ when comparing
ensembles whose samples have no one-to-one correspondence.
\end{definition}

\begin{proposition}[QFL invariance under a shared unitary]
\label{prop:qfl_invariance}
Let $\rho_i^\theta=U_\theta\rho_iU_\theta^\dagger$ for all samples in
$\mathcal{E}$. Then, for every $i,j$ and every $\theta$,
\begin{equation}
    F(\rho_i^\theta,\rho_j^\theta)=F(\rho_i,\rho_j).
    \label{eq:qfl_invariance}
\end{equation}
Consequently, optimizing $\theta$ cannot reshape the QFL induced by the prior ensemble.
\end{proposition}

The proposition does not imply that individual states remain unchanged: $U_\theta$ can
substantially move each state in Hilbert space. It states that their pairwise fidelity
relations move with them and remain fixed. The proof is given in \ref{app:qfl_invariance}.

\subsection{From QFL to Classical Output Separation}
QFL is defined in Hilbert space, whereas the discriminator receives classical images.
To connect the two spaces, write a fixed readout as $\mathcal{D}=h\circ\mathcal{M}$,
where the measurement $\mathcal{M}$ maps a state $\rho$ to a probability vector $p_\rho$
and $h$ is a fixed classical post-processing map.

\begin{assumption}[Stable fixed readout]
\label{ass:stable_readout}
Let $(\mathcal X,d_{\mathcal X})$ be the metric space of classical outputs, and let
$p=(p_k)_{k=1}^{K}$ and $q=(q_k)_{k=1}^{K}$ be probability vectors over the same
$K$ measurement outcomes. The post-processing map $h$ satisfies
\begin{equation}
 d_{\mathcal X}(h(p),h(q))\leq L\,\operatorname{TV}(p,q)
 \label{eq:stable_readout}
\end{equation}
for some finite $L$, where
$\operatorname{TV}(p,q):=\frac12\sum_{k=1}^{K}|p_k-q_k|$ is total variation distance
\citep{gibbs2002choosing}; equivalently, it is the classical trace distance of
\citet[Section~9.1]{nielsen2010quantum}.
\end{assumption}

This condition only excludes a readout that amplifies an arbitrarily small change in its
measurement probabilities into an unbounded output change. It does not prescribe a
particular measurement basis or image representation.

\paragraph{Scope with respect to classical decoders}
A fixed neural post-processing map is covered by Assumption~\ref{ass:stable_readout}
whenever it has a finite Lipschitz constant. A trainable family
$\{h_\phi\}_{\phi\in\Phi}$ is covered uniformly only if
$\sup_{\phi\in\Phi}\operatorname{Lip}(h_\phi)<\infty$. Without such a constraint, a
classical decoder may learn to amplify small readout differences, and no
parameter-independent output-separation bound follows from QFL alone. BasicQGAN avoids
this ambiguity by using a fixed, parameter-free decoder, allowing its metric distortion
to be characterized explicitly below.

\begin{theorem}[QFL-induced output-separation bound]
\label{thm:readout_bound}
Under Assumption~\ref{ass:stable_readout}, any pair of outputs obtained from a shared
unitary satisfies
\begin{equation}
 d_{\mathcal X}\!\left(\mathcal{D}(\rho_i^\theta),
                         \mathcal{D}(\rho_j^\theta)\right)
 \leq L\sqrt{1-F(\rho_i,\rho_j)}.
 \label{eq:general_output_bound}
\end{equation}
\end{theorem}

Thus, a high-fidelity prior pair cannot be mapped to arbitrarily distant outputs by the
shared unitary and a stable fixed readout. This is a one-sided constraint: low prior
fidelity permits a larger separation but does not guarantee that training will realize
it. In particular, QFL neither determines the generated distribution nor replaces the
need for an expressive circuit and successful adversarial optimization.

For Euclidean outputs $x_i=G_\theta(z_i)$ with
$\bar{x}=M^{-1}\sum_i x_i$, Theorem~\ref{thm:readout_bound} further gives
\begin{equation}
 \frac{1}{M}\sum_{i=1}^{M}\|x_i-\bar{x}\|_2^2
 \leq \frac{L^2}{2M^2}\sum_{i,j=1}^{M}
       \left(1-F(\rho_i,\rho_j)\right),
 \label{eq:output_dispersion_bound}
\end{equation}
which aggregates the pairwise constraint into an upper bound on output dispersion.
Proofs of Theorem~\ref{thm:readout_bound} and
Eq.~\eqref{eq:output_dispersion_bound} are deferred to
~\ref{app:readout_bound}.

\subsection{Amplitude-Readout Specialization}
The preceding result covers a general class of stable fixed readouts. We next specialize
the analysis to the amplitude representation used by BasicQGAN and obtain an explicit
bound without an unspecified readout constant.

\paragraph{Data-induced amplitude geometry}
For a nonnegative grayscale vector $x\in\mathbb{R}_+^d$, amplitude encoding prepares
\begin{equation}
 |\phi(x)\rangle=\sum_{k=1}^{d}\frac{x_k}{\|x\|_2}|k\rangle.
 \label{eq:data_amplitude_encoding}
\end{equation}
For two encoded images $x_i$ and $x_j$,
\begin{align}
 F_{ij}^{\mathrm{data}}
 &=\left(\frac{x_i^\top x_j}{\|x_i\|_2\|x_j\|_2}\right)^2, \\
 \left\|\frac{x_i}{\|x_i\|_2}-\frac{x_j}{\|x_j\|_2}\right\|_2^2
 &=2\left(1-\sqrt{F_{ij}^{\mathrm{data}}}\right).
 \label{eq:data_fidelity_geometry}
\end{align}
Hence, the data QFL describes the angular geometry of normalized image vectors under
this encoding.

\paragraph{Magnitude geometry of generated states}
Let
\begin{equation}
 |\psi_i\rangle=\sum_{k=1}^{d}c_{ik}|k\rangle,
 \qquad
 a_i=(|c_{i1}|,\ldots,|c_{id}|)^\top.
 \label{eq:magnitude_vector}
\end{equation}
The vector $a_i$ is the square-root probability representation of a computational-basis
measurement because $a_{ik}=\sqrt{p_{ik}}=|c_{ik}|$.

\begin{corollary}[Magnitude-distance bound]
\label{cor:magnitude_bound}
For any two pure states in Eq.~\eqref{eq:magnitude_vector},
\begin{equation}
 \|a_i-a_j\|_2^2
 \leq 2\left(1-\sqrt{F(\rho_i,\rho_j)}\right).
 \label{eq:magnitude_distance_bound}
\end{equation}
\end{corollary}

This result holds for both real and complex amplitudes because it depends on $|c_{ik}|$.
It also explains why we use square-root probabilities rather than probabilities as image
coordinates: square-root probabilities recover the amplitude magnitudes and admit the
direct fidelity--Euclidean relation in Eq.~\eqref{eq:magnitude_distance_bound}.
Probability readout remains a valid alternative, but it induces a different geometry.
The proof is provided in ~\ref{app:magnitude_bound}.

\paragraph{Fixed decoder used in BasicQGAN}
Let $\mathbf{1}\in\mathbb{R}^{d}$ be the all-ones vector and
$\|a\|_\infty:=\max_k|a_k|$. Ignoring the numerical stabilizer for clarity, the decoder is
\begin{equation}
 \widetilde{x}_i=2\frac{a_i}{\|a_i\|_\infty}-\mathbf{1}.
 \label{eq:basicqgan_decoder}
\end{equation}
It maps each sample to $[-1,1]^d$ before reshaping it into an image. The sample-wise
maximum normalization changes scale but preserves the normalized amplitude direction:
\begin{equation}
 \frac{\widetilde{x}_i+\mathbf{1}}
      {\|\widetilde{x}_i+\mathbf{1}\|_2}=a_i.
 \label{eq:decoder_direction}
\end{equation}
Therefore, Corollary~\ref{cor:magnitude_bound} applies directly to the normalized output
directions. Importantly, Eq.~\eqref{eq:decoder_direction} is an invertibility statement,
not an isometry claim: it shows that the amplitude-magnitude vector can be recovered
from the decoded sample, but it does not assert preservation of raw Euclidean distances.
For the decoder $T(a):=2a/\|a\|_\infty-\mathbf{1}$ on nonnegative unit vectors
$a,b\in\mathbb{R}_+^d$, ~\ref{app:magnitude_bound} proves the metric bounds
\begin{equation}
 \|a-b\|_2
 \leq \|T(a)-T(b)\|_2
 \leq 2(d+\sqrt d)\|a-b\|_2.
 \label{eq:decoder_metric_distortion}
\end{equation}
Consequently, the implemented normalization may distort pairwise Euclidean distances,
but only by fixed dimension-dependent factors; together with
Eq.~\eqref{eq:magnitude_distance_bound}, it yields the one-sided bound
\begin{equation}
 \|T(a_i)-T(a_j)\|_2
 \leq 2(d+\sqrt d)\sqrt{2\left(1-\sqrt{F(\rho_i,\rho_j)}\right)}.
 \label{eq:decoded_fidelity_bound}
\end{equation}
Thus, QFL constrains rather than reconstructs the Euclidean geometry seen after decoding.
The implication remains one-sided because the magnitude readout discards amplitude
signs and phases: two low-fidelity states can still yield the same magnitude vector and
hence the same decoded image. QFL calibration should therefore be interpreted as a
structural prior for the subsequent adversarial optimization, rather than as a guarantee
of Euclidean-distance matching or distribution recovery.
In our implementation, both state preparation and the generator use only
$R_Y$ and CNOT gates, whose ideal matrices are real. The statevector is consequently
real-valued, and the implemented quantity $|\operatorname{Re}(c_k)|$ equals
$|c_k|=\sqrt{p_k}$ up to numerical precision. The theoretical formulation in terms of
$|c_k|$ is more general: if complex-valued gates are used, the same decoder should be
implemented as $\sqrt{p_k}$ rather than by retaining only the real part.

\subsection{QFL-Guided Quantum Prior Calibration}
Proposition~\ref{prop:qfl_invariance} shows that the prior QFL cannot be corrected by
updating the shared unitary during adversarial training. Equations~\eqref{eq:data_fidelity_geometry}
and~\eqref{eq:magnitude_distance_bound} further show that, under the amplitude representation,
this invariant is connected to the geometry of the decoded samples. These observations
motivate calibrating the prior before adversarial training.

Because prior and data samples have no pairwise correspondence, BasicQGAN does not match
two QFL matrices entry by entry. For an ensemble $\mathcal{E}$, we instead define the
empirical distribution of off-diagonal fidelities as
\begin{equation}
\widehat{\mu}_{\mathcal{E}}
 =\frac{2}{M(M-1)}\sum_{1\leq i<j\leq M}
   \delta_{F(\rho_i,\rho_j)}.
\label{eq:empirical_qfl}
\end{equation}
Here $\delta_t$ denotes the Dirac point mass located at $t\in[0,1]$.
Real images are amplitude encoded to construct
$\widehat{\mu}_{\mathrm{data}}$. A parameterized state-preparation distribution with
parameters $\alpha=(a_1,\ldots,a_N)$ induces
$\widehat{\mu}_{\mathrm{prior}}(\alpha)$.

For qubit $i$, we reparameterize the sampled rotation as
$Z_i=\pi a_i u_i$, where $u_i\sim U[0,1]$. This is equivalent to sampling from
$U[0,a_i\pi]$ when $a_i\geq0$ and from $U[a_i\pi,0]$ otherwise. We optimize
\begin{equation}
 \alpha^*=\arg\min_\alpha
 \mathcal{L}_{\mathrm{QFL}}(\alpha),
 \label{eq:qfl_calibration_problem}
\end{equation}
where
\begin{align}
 \mathcal{L}_{\mathrm{QFL}}
 ={}&D_{\mathrm{KDE}}(\widehat{\mu}_{\mathrm{prior}},
                      \widehat{\mu}_{\mathrm{data}})
 +\left|\widehat{m}_{\mathrm{prior}}-\widehat{m}_{\mathrm{data}}\right| \nonumber\\
 &+\left|\widehat{v}_{\mathrm{prior}}-\widehat{v}_{\mathrm{data}}\right|.
 \label{eq:qfl_loss}
\end{align}
Here $\widehat{m}$ and $\widehat{v}$ are the empirical mean and variance of pairwise
fidelities. We evaluate both kernel-density estimates on a common grid
$\{q_b\}_{b=1}^{B}$ and use the discretized absolute discrepancy
\begin{equation}
 D_{\mathrm{KDE}}(\mu,\nu)
 =\frac{1}{B}\sum_{b=1}^{B}
 |\widehat{f}_{\mu}(q_b)-\widehat{f}_{\nu}(q_b)|.
 \label{eq:kde_discrepancy}
\end{equation}
For the empirical QFL distributions induced by finite state ensembles, we combine a
smoothed density discrepancy with explicit moment constraints. The KDE term compares
the overall density profiles while retaining sensitivity to differences in local peaks
and low-density regions. The mean and variance terms explicitly constrain the overall
fidelity level and the degree of concentration, respectively. By jointly matching
distributional shape, location, and dispersion, this composite objective provides a
statistically interpretable optimization signal for quantum-prior calibration. It can
be viewed as one finite-sample surrogate for QFL calibration; other statistical measures
of discrepancy between the prior and data QFLs may also be explored within this
framework.
Algorithm~\ref{alg:Optimization algorithm} summarizes the offline stage.

\begin{algorithm}
\caption{QFL-guided quantum prior calibration}
\label{alg:Optimization algorithm}
\textbf{Input}: Number of epochs $n_{\mathrm{epoch}}$, number of qubits $N$,
ensemble size $M$, learning rate $\eta$, data QFL distribution
$\widehat{\mu}_{\mathrm{data}}$, and trainable prior parameters $\alpha$.
\begin{algorithmic}[1]
\STATE Initialize $\alpha=(a_1,\ldots,a_N)$.
\FOR{$t=1,\ldots,n_{\mathrm{epoch}}$}
\STATE Sample $u_i^{(m)}\sim U[0,1]$ and set
$Z_i^{(m)}=\pi a_i u_i^{(m)}$ for $m=1,\ldots,M$.
\STATE Prepare $M$ initial states and compute all pairwise fidelities.
\STATE Construct $\widehat{\mu}_{\mathrm{prior}}(\alpha)$ using
Eq.~\eqref{eq:empirical_qfl}.
\STATE Evaluate $\mathcal{L}_{\mathrm{QFL}}$ using Eq.~\eqref{eq:qfl_loss}.
\STATE Update $\alpha\leftarrow\operatorname{Adam}
(\mathcal{L}_{\mathrm{QFL}},\alpha,\eta)$.
\ENDFOR
\STATE Return the calibrated sampling parameters $\alpha^*$.
\end{algorithmic}
\end{algorithm}

This objective matches distributional statistics rather than asserting exact equality of
the two finite QFL matrices. Calibration is performed once, and $\alpha^*$ remains fixed
during online adversarial training. Thus, QFL calibration supplies a prior geometry; it
does not replace the adversarial objective or guarantee successful generation by itself.

\paragraph{Consistency of finite-ensemble estimation}
The $\binom{M}{2}$ pairwise fidelities in Eq.~\eqref{eq:empirical_qfl} are dependent
because each state participates in multiple pairs. If
$\rho_1,\ldots,\rho_M$ are i.i.d. states sampled under a fixed prior parameter $\alpha$,
then, for any bounded measurable $\varphi:[0,1]\to\mathbb{R}$,
\begin{equation}
 \frac{2}{M(M-1)}\sum_{i<j}\varphi(F(\rho_i,\rho_j))
 \xrightarrow{\mathrm{a.s.}}
 \mathbb{E}[\varphi(F(\rho_1,\rho_2))].
 \label{eq:qfl_ustat_convergence}
\end{equation}
The left-hand side is an order-two U-statistic. By choosing $\varphi$ for the first and
second moments or for a fixed-bandwidth KDE evaluated on a fixed finite grid,
Eq.~\eqref{eq:qfl_ustat_convergence} establishes consistency of the empirical QFL
descriptors, including its mean, variance, and KDE values. Consequently, as $M$
increases, the empirical QFL histogram converges to the corresponding population
histogram when a common bin partition is used across ensemble sizes. The
composite loss in Eq.~\eqref{eq:qfl_loss} is computed from these consistent QFL
estimates. ~\ref{app:qfl_stabilization} provides the complete proof.
Figure~\ref{statistic} compares QFL estimates obtained from 500, 1000, 1500, and 2000
states for digit~0 at
$16\times16$ resolution. Based on the observed trade-off between estimation stability
and computational cost, we use 1500 states for the remaining QFL calculations.

Figure~\ref{Hypothesis} compares the data QFL with priors produced by several sampling
ranges. The broad priors $U[0,\pi]$ and $U[0,1]$ and the highly concentrated setting
$Z_0\sim U[0,0.01]$, $Z_i=0$ for $i>0$, all yield visibly different QFL distributions.
The calibrated prior more closely matches the data-QFL statistics. Corresponding results
for other datasets are reported in ~\ref{app_3}.

\begin{figure}[H]
    \centering
    \includegraphics[width=0.8\linewidth]{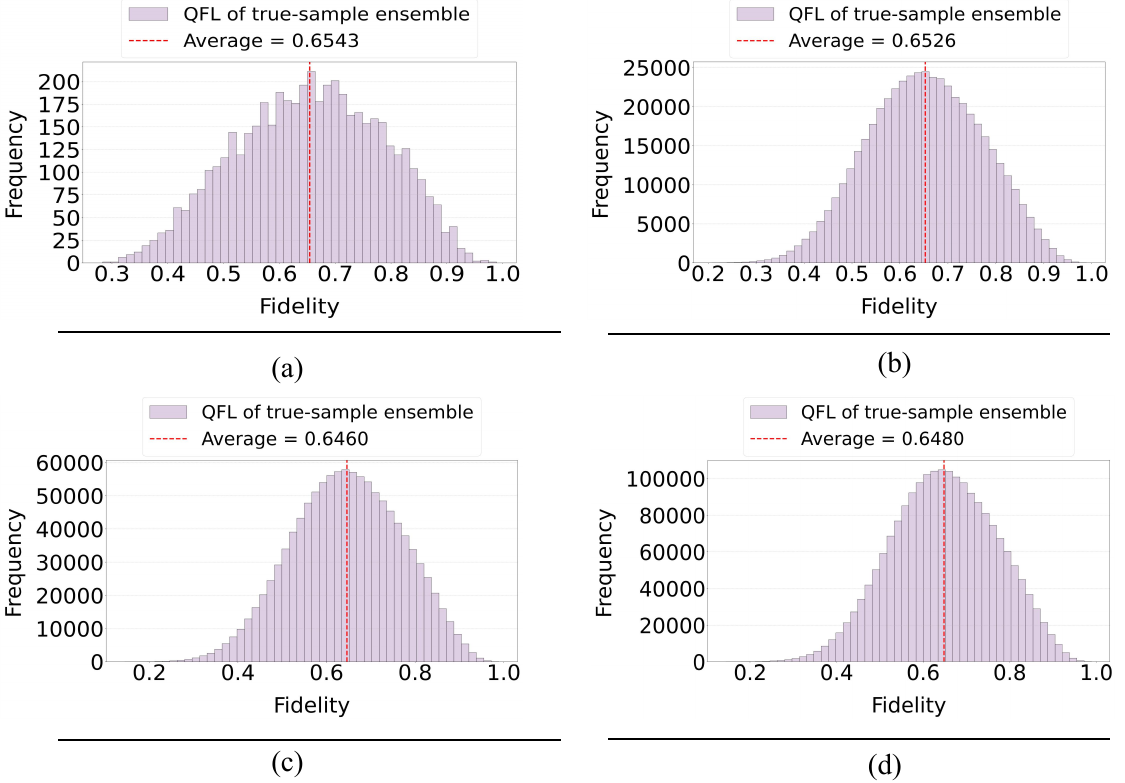}
    \caption{Empirical QFL distributions computed from 500, 1000, 1500, and 2000
    real samples of MNIST digit~0 at $16\times16$ resolution.}
    \label{statistic}
\end{figure}

\begin{figure}[H]
    \centering
    \includegraphics[width=0.8\linewidth]{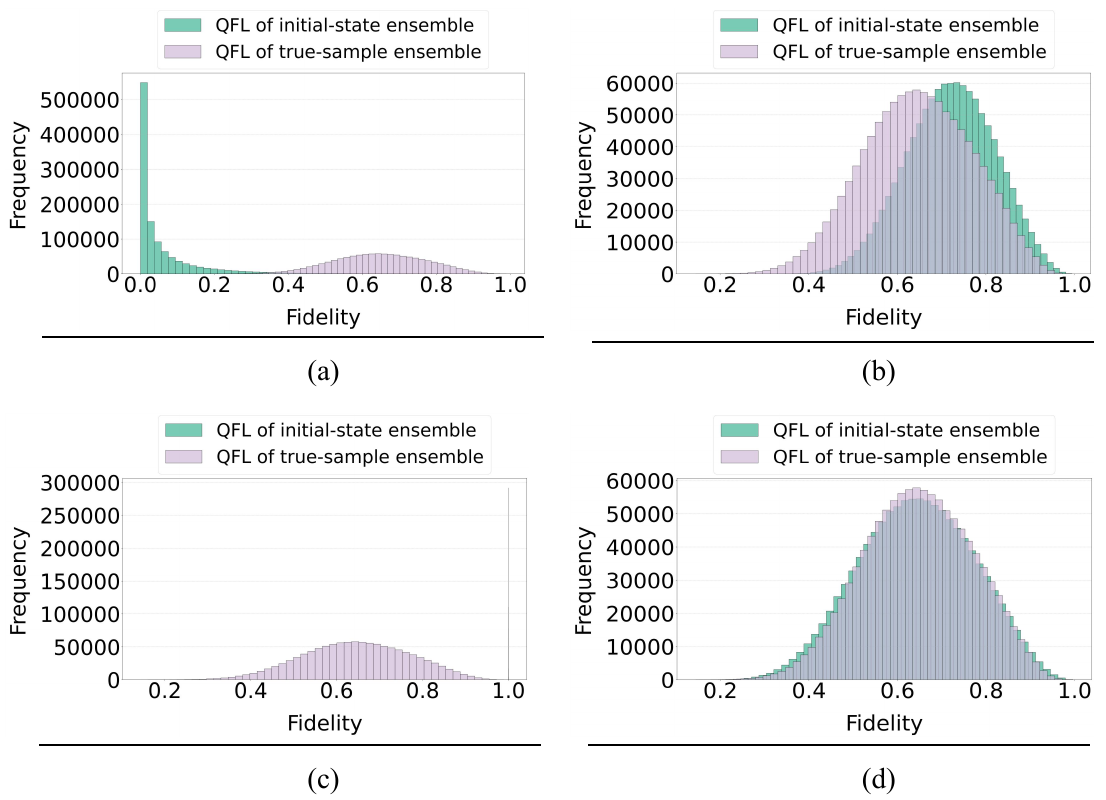}
    \caption{QFL distributions induced by different prior preparations and by real data.
    Panels (a)--(c) use uncalibrated sampling schemes, whereas panel (d) uses the
    QFL-calibrated prior.}
    \label{Hypothesis}
\end{figure}

\subsection{Single-Circuit Generation and Online Training}
The BasicQGAN generator has $N=\log_2 d$ qubits for a $d$-pixel image. The calibrated
rotation distribution first prepares
\begin{equation}
 |\psi_z^0\rangle=
 \left[\bigotimes_{i=1}^{N}R_Y(Z_i)\right]|0\rangle^{\otimes N}.
 \label{eq:prior_state}
\end{equation}
The shared generator then applies $L$ layers of trainable $R_Y$ rotations and nearest-neighbor CNOT gates:
\begin{equation}
 |\psi_z^\theta\rangle=U_\theta|\psi_z^0\rangle
 =\sum_{k=1}^{d}c_k(z,\theta)|k\rangle.
 \label{eq:generated_state}
\end{equation}
The fixed decoder in Eq.~\eqref{eq:basicqgan_decoder} converts this state into the complete
pixel vector; no patch-wise generation or trainable classical decoder is used.

We use the same classical discriminator architecture as PQWGAN and optimize the generator
and discriminator with WGAN-GP. Denoting the critic by $D_\omega$, the objective is
\begin{align}
 \min_\theta\max_\omega\quad
 &\mathbb{E}_{x\sim P_{\mathrm{data}}}[D_\omega(x)]
 -\mathbb{E}_{z\sim P_{\alpha^*}}[D_\omega(G_\theta(z))] \nonumber\\
 &-\lambda_{\mathrm{GP}}\mathbb{E}_{\widehat{x}}
 \left(\|\nabla_{\widehat{x}}D_\omega(\widehat{x})\|_2-1\right)^2.
 \label{eq:wgan_gp}
\end{align}
The offline stage optimizes $\alpha$ to calibrate the prior QFL, whereas the online stage
holds $\alpha^*$ fixed and optimizes $\theta$ and $\omega$ to learn the image distribution.
The complete alternating optimization procedure is given in \ref{app_4}.

% ---- END INLINED METHOD SECTION ----

\section{Experiments}
\paragraph{Datasets}
\label{Datasets and Evaluation Metric}
%Given current limitations in quantum resources and technology, QGANs employing pure quantum generators are most appropriate for processing small-scale grayscale image datasets. Therefore, we constructed four such datasets based on three public handwritten image repositories: MNIST \cite{deng2012mnist}, EMNIST \cite{cohen2017emnistextensionmnisthandwritten}, and HDS \cite{robert2022hds}. These datasets include two MNIST versions (original 28$\times$28 and downsampled 16$\times$16, both with 10 classes), a 16$\times$16 Geometric Shapes dataset (4 classes), and a 16$\times$16 Uppercase Letters dataset (10 classes). More dataset details are provided in Appendix \ref{app_5}.
% Constrained by current quantum hardware limitations, pure quantum generators are best suited for small-scale grayscale images. 
% Accordingly, we curated four datasets derived from three public repositories: MNIST \cite{deng2012mnist}, EMNIST \cite{cohen2017emnistextensionmnisthandwritten}, and HDS \cite{robert2022hds}. These benchmarks include two MNIST variants (original 28$\times$28 and downsampled 16$\times$16, 10 classes), Geometric Shapes (16$\times$16, 4 classes), and Uppercase Letters (16$\times$16, 10 classes). Further details are provided in Appendix \ref{app_5}.

Given the current scale of quantum resources, we focus on small-scale grayscale image benchmarks. Accordingly, we curated four datasets derived from three public repositories: MNIST~\citep{deng2012mnist}, EMNIST~\citep{cohen2017emnistextensionmnisthandwritten}, and HDS~\citep{robert2022hds}. These benchmarks include two MNIST variants (original $28\times28$ and downsampled $16\times16$, 10 classes), Geometric Shapes ($16\times16$, 4 classes), and Uppercase Letters ($16\times16$, 10 classes). Further details are provided in \ref{app_5}.

\paragraph{Baselines and Evaluation Metric}

Our objective is direct end-to-end, pixel-level synthesis of full images using a single quantum circuit, and we therefore choose PQWGAN~\citep{tsang2023hybrid} as the baseline.
With \texttt{Patch=1} (i.e., without patch-wise decomposition), PQWGAN can be viewed as a naive end-to-end variant of the prevailing patch-based QGAN paradigm, enabling a direct and controlled comparison to assess whether our architecture alleviates the inherent limitations of conventional designs. We denote the \texttt{Patch=1} configuration of PQWGAN as PQWGAN(Global).
In addition,
we also benchmark our method against the patch-based quantum model PQWGAN and the classical model WGAN-GP. PQWGAN uses a fully quantum generator and represents the current state-of-the-art for directly modeling pixel data. Classical WGAN-GP can be considered a classical analogue of BasicQGAN. Generated images are evaluated using the Fr\'echet Inception Distance (FID) \citep{yang2025ihqgan,karras2019style}, a standard metric for visual quality and diversity, where a lower score indicates better performance.
% We evaluate the quality and diversity of the generated images using the Frechet Inception Distance (FID). A lower FID score indicates that the generated images are closer to real images in both visual quality and diversity. 
The FID is calculated using formula
$F I D=\left\|\mu_r-\mu_g\right\|^2+\operatorname{Tr}\left(\Sigma_r+\Sigma_g-2\left(\Sigma_r \Sigma_g\right)^{1 / 2}\right)$,
where $\mu_r$ and $\mu_g$ represent the feature means of the real images and the generated images, respectively. $\Sigma_r$ and $\Sigma_g$ denote the feature covariance matrices of the real and generated images, respectively. $\text{Tr}$ means the matrix trace.

\paragraph{Implementation Details}
All experiments in this study were simulated using PyTorch and PennyLane on a 13th Gen Intel(R) Core(TM) i5-13490F CPU with 32.0 GB of RAM. The initial-state ensemble preparation algorithm was trained using Adam with an initial learning rate of 0.3. During adversarial training, both the classical critic and quantum generator were optimized using Adam with initial learning rates of 0.0002 and 0.01, respectively, and momentum hyperparameters of 0 and 0.9. All models were trained for 50 epochs with a batch size of 5.

\section{Results}

We organize the empirical evidence around the paper's main claim. First, we test whether a naive global QGAN baseline fails in the single-circuit end-to-end setting. Second, we evaluate whether QFL calibration restores stable training and improves generation quality. Third, we compare against patch-based and classical baselines across resolutions and datasets. Finally, we examine resource overhead and ablations to clarify where the proposed prior calibration helps and where it remains limited.

\subsection{Naive End-to-End Baseline}
As shown in Figure~\ref{Fid_3_1}, we compare BasicQGAN with PQWGAN(Global) on MNIST at $16\times16$.
PQWGAN(Global) removes patch-wise decomposition but keeps a conventional prior, making it a controlled test of naive single-circuit generation. BasicQGAN and PQWGAN(Global) use the same quantum generator circuit, so the comparison isolates the effect of QFL-guided prior calibration. PQWGAN(Global) produces blurrier samples with less distinguishable digit contours, while BasicQGAN yields clearer and more stable structures.
As reported in Figure~\ref{Fid_B_VS_PQWGAN}, BasicQGAN also achieves lower FID. This comparison supports our diagnosis that the main difficulty is not merely the use of one circuit, but the mismatch between the initial quantum prior and the data-induced fidelity structure.

\begin{figure}[H]
    \centering
   \includegraphics[width=\linewidth]{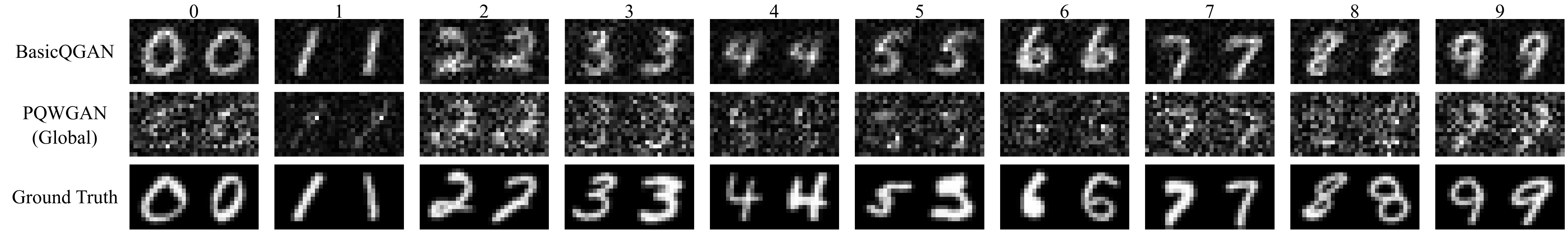}
    \caption{BasicQGAN produces clearer MNIST samples than the end-to-end PQWGAN(Global) baseline at $16\times16$ resolution.}
    \label{Fid_3_1}
\end{figure}

\begin{figure}[H]
    \centering
   \includegraphics[width=0.8\linewidth]{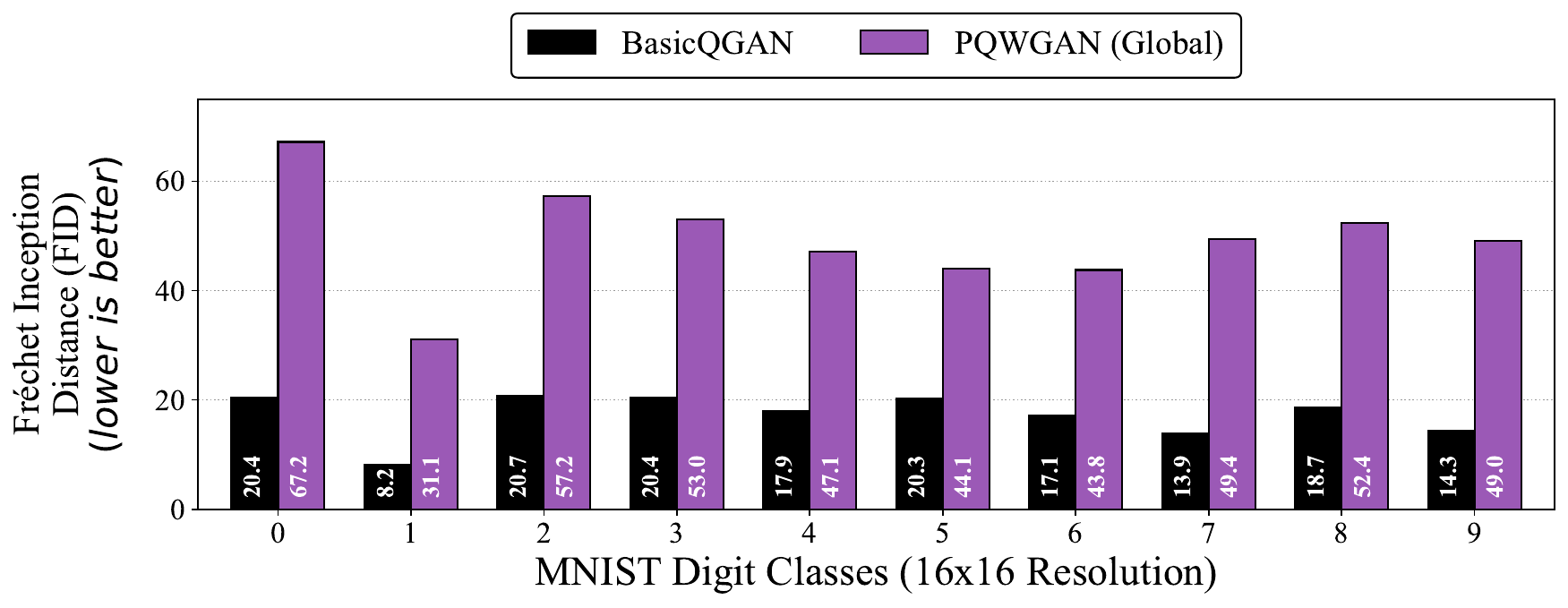}
    \caption{BasicQGAN obtains lower FID scores than PQWGAN(Global) across MNIST classes at $16\times16$ resolution.}
    \label{Fid_B_VS_PQWGAN}
\end{figure}

% \subsection{Effect of QFL Calibration}
% The quantum-prior ablation tests whether QFL calibration is associated with the observed improvement. As shown in Figure~\ref{crop_ablation_study6}, we replace the calibrated prior with $Z_i\!\sim\!U[0,\pi]$, $Z_i\!\sim\!U[0,1]$, and an extreme narrow setting $Z_0\!\sim\!U[0,0.01]$. These alternatives break QFL alignment and lead to either complete training failure or severe mode collapse. In contrast, our initial-state ensemble preparation algorithm enables stable training and produces diverse samples, achieving the lowest FID of 20.38 in Table~\ref{tab:sampling interval}. These results provide empirical support for the diagnosis--method link in the tested fixed-readout architecture: the evaluated mismatched priors fail, whereas QFL calibration mitigates this failure.

\subsection{Impact of Prior QFL on Generation Performance}
The quantum-prior ablation examines how different prior-induced QFLs affect generation performance.  
% The quantum-prior ablation tests whether QFL calibration is associated with the observed improvement. 
As shown in Figure~\ref{crop_ablation_study6}, we replace the calibrated prior with $Z_i\!\sim\!U[0,\pi]$, $Z_i\!\sim\!U[0,1]$, and an extreme narrow setting $Z_0\!\sim\!U[0,0.01]$. The QFLs induced by these priors have been shown earlier in Figure~\ref{Hypothesis}; all of them deviate visibly from the real-data QFL. Consistently, the broad priors lead to complete training failure in Figure~\ref{crop_ablation_study6}(a) and Figure~\ref{crop_ablation_study6}(b), while the overly narrow prior results in severe mode collapse in Figure~\ref{crop_ablation_study6}(c). This indicates that generation training becomes difficult when the prior-induced QFL differs substantially from the real-data QFL. In particular, when the prior sampling range is too narrow, the generated quantum states become overly concentrated in Hilbert space, making the subsequent adversarial training prone to mode collapse. In contrast, our initial-state ensemble preparation algorithm calibrates the quantum prior and enables the subsequent QGAN training to generate diverse samples, as shown in Figure~\ref{crop_ablation_study6}(d), achieving the lowest FID of 20.38 in Table~\ref{tab:sampling interval}. 
% These results provide empirical support for the diagnosis--method link in the tested fixed-readout architecture: the evaluated mismatched priors fail, whereas QFL calibration mitigates this failure.

\begin{figure}[H]
    \centering
   \includegraphics[width=0.68\linewidth]{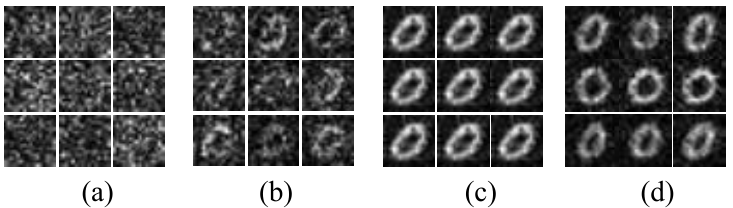}
    \caption{Ablation study on the initial-state preparation algorithm for BasicQGAN on the 16$\times$16 digit '0' generation task. (a) and (b) show complete training failure resulting from sampling initial parameters from wide intervals, $U[0, \pi]$ and $U[0, 1]$, respectively. (c) demonstrates severe mode collapse when using an extremely narrow interval, $U[0, 0.01]$. In contrast, (d) shows that our proposed algorithm successfully generates a diverse set of high-quality images. 
    % demonstrating that our proposed preparation algorithm is crucial for the successful operation of BasicQGAN.
    }
    \label{crop_ablation_study6}
\end{figure}

\begin{table}[H]
    %\caption{Quantitative comparison of the full BasicQGAN against its variants, where our initial-state ensemble preparation method is replaced with baseline schemes.}
    \caption{Quantitative comparison of BasicQGAN against variants with different initial-state ensemble preparation methods.}
    \label{tab:sampling interval} 
    \small
    \centering
\begin{tabular}{@{}p{7cm} c@{}} 
        \toprule
        Model & FID\\
        \midrule 
        BasicQGAN with $Z_i \sim U(0, \pi)$ & 73.93\\
        BasicQGAN with $Z_i \sim U(0, 1)$ & 72.64\\
        \makecell[l]{BasicQGAN $Z_0 \sim U(0, 0.01)$, $Z_i = 0 \, (i>0)$} & 33.97\\
        BasicQGAN (Our algorithm)  & \textbf{20.38}\\  
        \bottomrule 
    \end{tabular}
    \captionsetup{font=normalsize} 
\end{table}
%\vspace{-1ex}
\normalsize

\subsection{Multi-Scale Performance}
We evaluate different methods at two resolutions. As shown in Figures~\ref{crop_MNIST}--\ref{FID_scores_28x28_with_labels2}, BasicQGAN and WGAN-GP produce relatively natural and smooth samples at $16\times16$ resolution, whereas PQWGAN generates visually stiffer digits and yields higher FID. BasicQGAN already produces images close to those of the classical WGAN-GP baseline, indicating that single-circuit end-to-end generation is feasible at lower resolution. 
% The quantitative results for both resolutions are summarized in Table~\ref{tab:FID_scores}.
At $28\times28$ resolution, BasicQGAN exhibits more visible noise in its generated samples. The FID scores of the three methods are generally close, with BasicQGAN slightly higher than the other baselines; meanwhile, PQWGAN still shows noticeable discrete artifacts in its generated samples. We attribute the degradation of BasicQGAN at $28\times28$ mainly to two factors: (1) the significant increase in data dimensionality and computational complexity, which enlarges the QFL matching error; and (2) the increased complexity of the image distribution, which may exceed the expressive capacity of the current quantum circuit. Thus, the $28\times28$ results further show the feasibility of single-circuit end-to-end generation in a higher-dimensional image space, while also leaving room for improving generation quality at higher resolutions.

\begin{figure}[H]
    \centering
    \begin{minipage}[c]{0.68\textwidth}
        \centering
        \includegraphics[width=\linewidth]{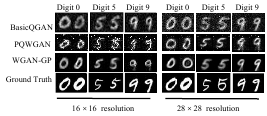}
    \end{minipage}
    \hspace{0.035\textwidth}
    \begin{minipage}[c]{0.22\textwidth}
        \centering
        \includegraphics[width=\linewidth]{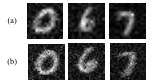}
    \end{minipage}
    \caption{Visual comparison of generated MNIST samples at $16\times16$ and $28\times28$ resolutions. Left: BasicQGAN and WGAN-GP produce relatively smooth and natural images at both resolutions, whereas the images generated by PQWGAN are blurrier and less natural, particularly for digits `5' and `9'. Right: magnified BasicQGAN samples at $16\times16$ resolution (a) and $28\times28$ resolution (b), where the latter exhibits increased noise artifacts.}
    \label{crop_MNIST}
    \label{crop_16_28}
\end{figure}

\begin{figure}[H]
    \centering
   \includegraphics[width=0.8\textwidth, trim=0cm 0cm 0cm 0cm, clip]{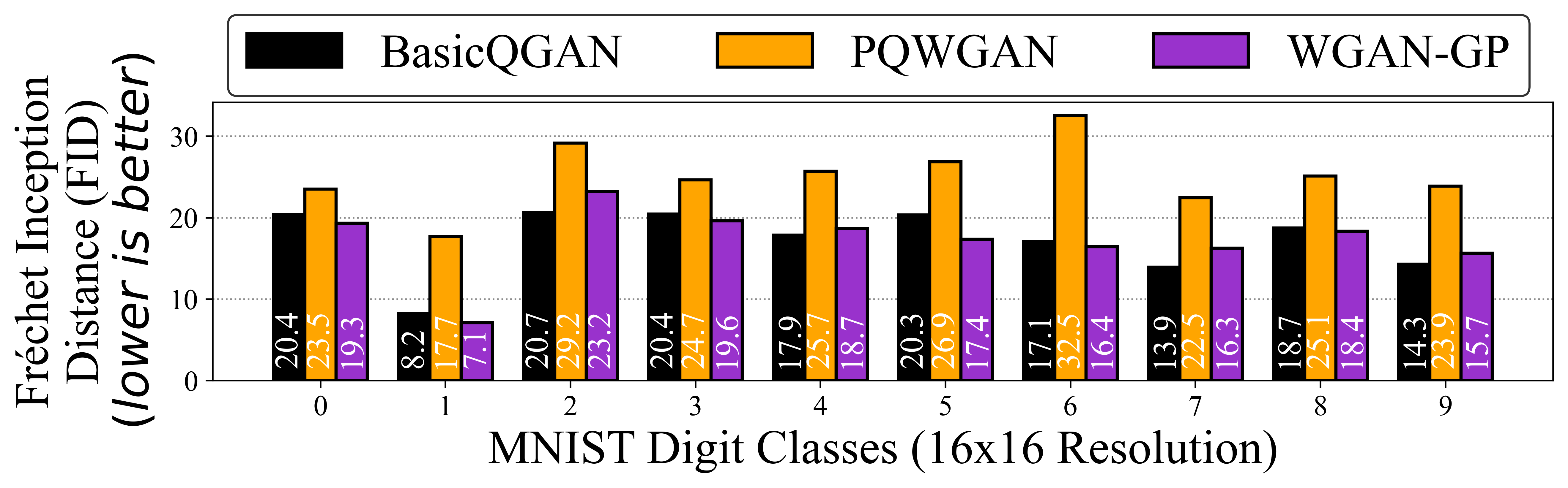}
    \caption{
 WGAN-GP and BasicQGAN achieve similar FID scores across different classes of the 16$\times$16 MNIST dataset, and both consistently obtain lower scores than PQWGAN.
    }
    \label{FID_scores_16x16_with_labels2}
\end{figure}

\begin{figure}[H]
    \centering
   \includegraphics[width=0.8\textwidth, trim=0cm 0cm 0cm 0cm, clip]{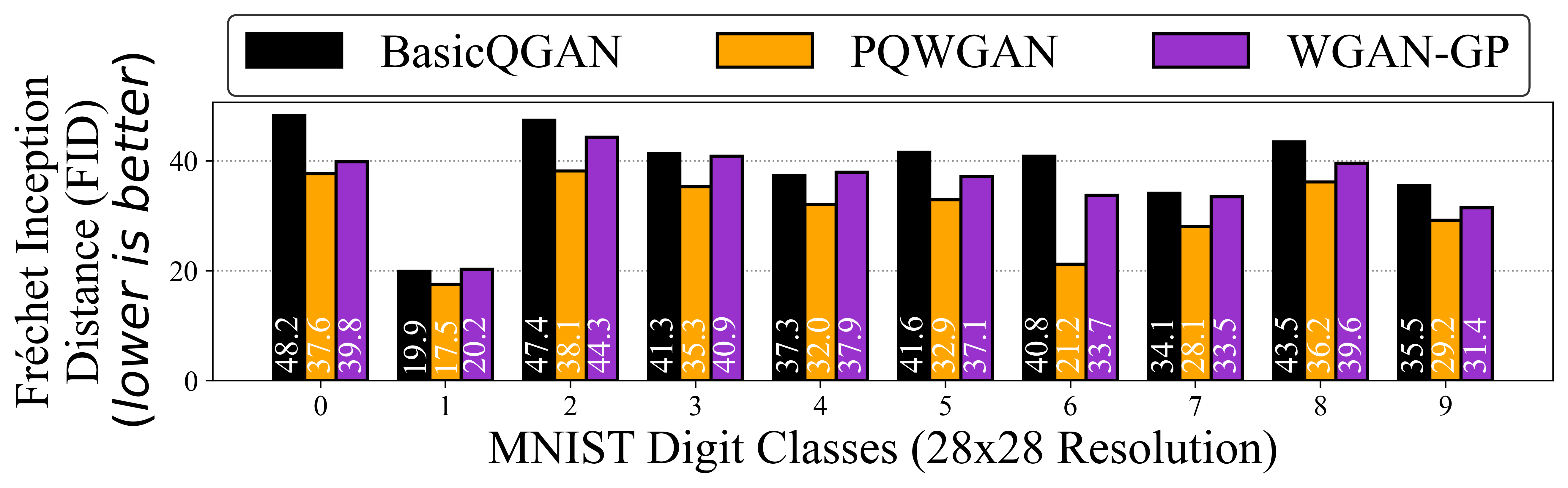}
    \caption{ 
    At $28\times28$ resolution, BasicQGAN obtains higher FID scores than the baselines, while PQWGAN achieves the lowest FID scores among the compared methods.
    }
    \label{FID_scores_28x28_with_labels2}
\end{figure}

% \begin{table}[H]
% \caption{FID scores on MNIST across digit classes under different resolutions. Lower is better.}
% \label{tab:FID_scores}
% \centering
% \small
% \setlength{\tabcolsep}{2.5pt}
% \renewcommand{\arraystretch}{1.1}
% \begin{tabular}{c|ccc|ccc}
% \toprule
% & \multicolumn{3}{c|}{\textbf{16$\times$16}} 
% & \multicolumn{3}{c}{\textbf{28$\times$28}} \\
% \cmidrule(lr){2-4} \cmidrule(lr){5-7}
% Digit 
% & \makecell{Basic\\QGAN} 
% & \makecell{PQ\\WGAN} 
% & \makecell{WGAN\\-GP} 
% & \makecell{Basic\\QGAN} 
% & \makecell{PQ\\WGAN} 
% & \makecell{WGAN\\-GP} \\
% \midrule
% 0 & 20.4 & 23.5 & \textbf{19.3} & 48.2 & \textbf{37.6} & 39.8 \\
% 1 &  8.2 & 17.7 & \textbf{7.1}  & 19.9 & \textbf{17.5} & 20.2 \\
% 2 & \textbf{20.7} & 29.2 & 23.2 & 47.4 & \textbf{38.1} & 44.3 \\
% 3 & 20.4 & 24.7 & \textbf{19.6} & 41.3 & \textbf{35.3} & 40.9 \\
% 4 & \textbf{17.9} & 25.7 & 18.7 & 37.3 & \textbf{32.0} & 37.9 \\
% 5 & 20.3 & 26.9 & \textbf{17.4} & 41.6 & \textbf{32.9} & 37.1 \\
% 6 & 17.1 & 32.5 & \textbf{16.4} & 40.8 & \textbf{21.2} & 33.7 \\
% 7 & \textbf{13.9} & 22.5 & 16.3 & 34.1 & \textbf{28.1} & 33.5 \\
% 8 & 18.7 & 25.1 & \textbf{18.4} & 43.5 & \textbf{36.2} & 39.6 \\
% 9 & \textbf{14.3} & 23.9 & 15.7 & 35.5 & \textbf{29.2} & 31.4 \\
% \bottomrule
% \end{tabular}
% \end{table}

\subsection{Cross-Dataset Generalization}
We evaluate the cross-dataset generalization ability of different quantum models on two additional $16\times16$ datasets: Geometric Shapes and Uppercase Letters. As shown in Figure~\ref{crop_Geo_Uppercase}, BasicQGAN generates relatively smooth and natural samples on both datasets. In contrast, PQWGAN produces blurrier and less natural samples. Quantitatively, Figures~\ref{Fid_Uppecase} and~\ref{Fid_Geometric} show that BasicQGAN achieves the lowest FID on both datasets.

\begin{figure}[H]
    \centering
   \includegraphics[width=0.8\textwidth]{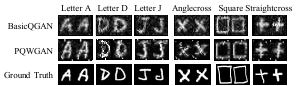}
    % \caption{Quantum model performance comparison on the Uppercase Letters and Geometric Shapes datasets. BasicQGAN produces smoother samples, while PQWGAN exhibits noticeable blurriness, edge artifacts, and discrete white noise in several outputs, such as the uppercase letters `D' and `J'.}
    \caption{Quantum model performance comparison on the Uppercase Letters and Geometric Shapes datasets. BasicQGAN produces smoother and more natural samples, whereas PQWGAN generates blurrier and less natural outputs.}
    \label{crop_Geo_Uppercase}
\end{figure}

\begin{figure}[H]
    \centering
   \includegraphics[width=0.8\textwidth, trim=0cm 0cm 0cm 0cm, clip]{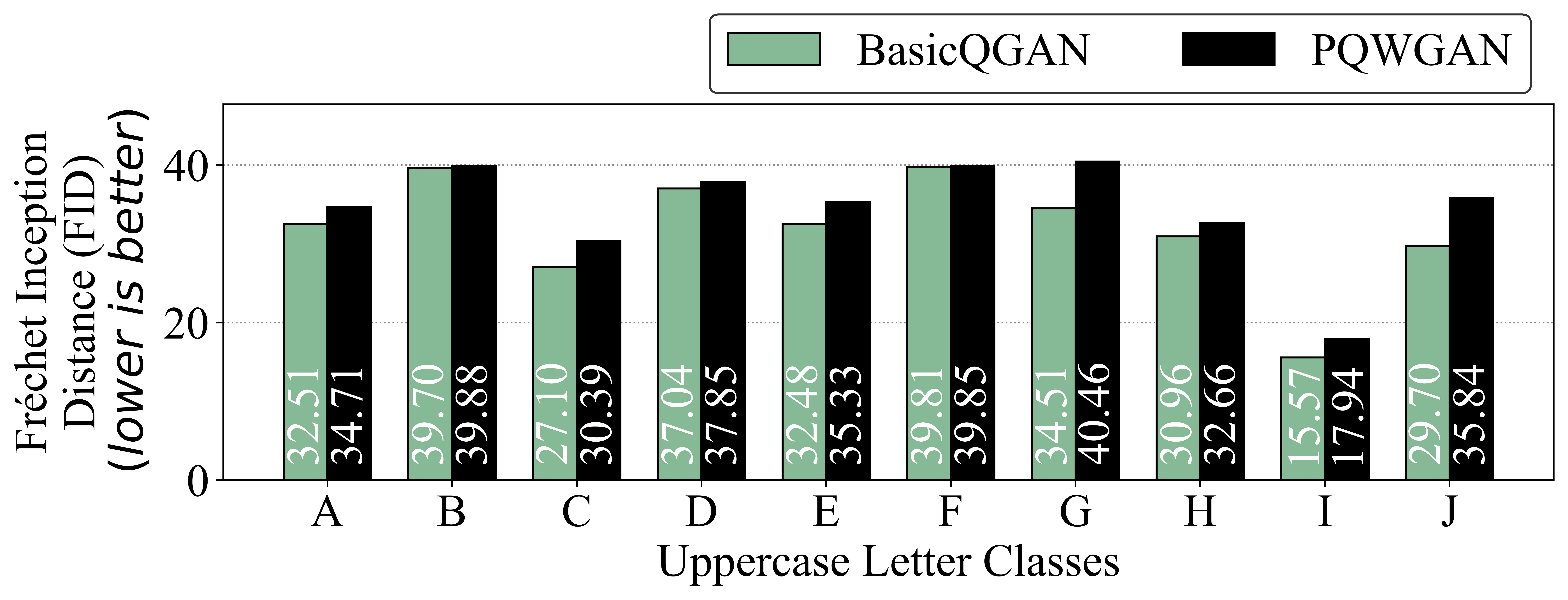}
    \caption{BasicQGAN obtains lower FID scores than PQWGAN across different classes of the Uppercase Letter dataset.}
    \label{Fid_Uppecase}
\end{figure}

\begin{figure}[H]
    \centering
   \includegraphics[width=0.8\textwidth, trim=0cm 0cm 0cm 0cm, clip]{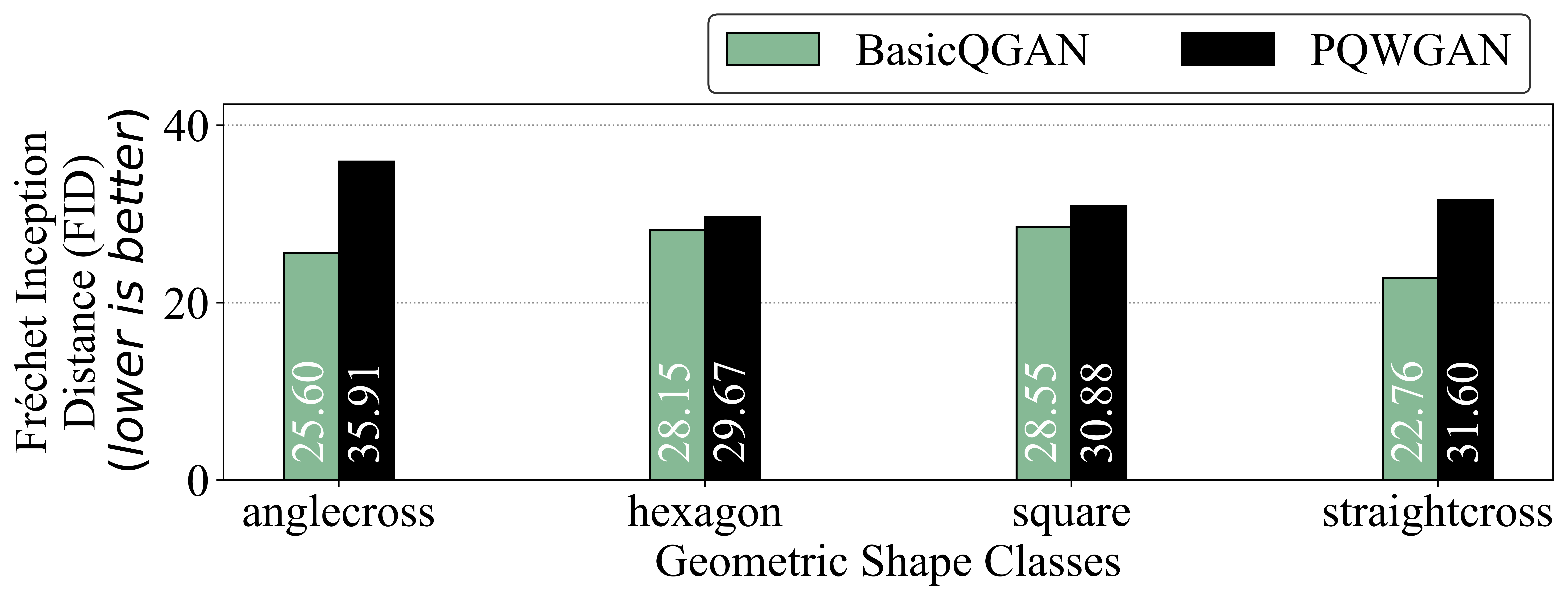}
    \caption{BasicQGAN obtains lower FID scores than PQWGAN across different classes of the Geometric Shape dataset.}
    \label{Fid_Geometric}
\end{figure}

\subsection{Resource Efficiency}
Table~\ref{tab:generator_parameter} summarizes the generator resource overhead of the compared methods at $16\times16$ and $28\times28$ resolutions. Among the quantum models, BasicQGAN consistently requires fewer qubits and trainable parameters than both PQWGAN(Global) and the patch-based PQWGAN. The difference becomes more pronounced at the higher resolution, where the resource demand of patch-based PQWGAN increases substantially. Compared with the classical WGAN-GP, BasicQGAN also uses a much smaller number of trainable parameters. Overall, BasicQGAN maintains a compact generator architecture across both resolutions.

\begin{table}[H]
    \caption{
    Comparison of the resource overhead of the generators from different methods on the MNIST dataset at $16\times16$ and $28\times28$ resolutions.}
    \label{tab:generator_parameter}
    \centering
    \small
    \setlength{\tabcolsep}{1pt} 
    \begin{tabular}{llcccc} 
    \toprule
    Method & Type & \multicolumn{2}{c}{\makecell{ MNIST Dataset \\16$\times$16}} & \multicolumn{2}{c}{\makecell{MNIST Dataset \\28$\times$28}} \\
    \cmidrule(lr){3-4} \cmidrule(lr){5-6} 
        & & \makecell{Qubit\\Count} & \makecell{Parameter\\Count} &   \makecell{Qubit\\Count} &\makecell{Parameter\\Count} \\ 
        \midrule
        PQWGAN(Global) & Quantum & 9 & 540 & 11 & 660 \\ 
        BasicQGAN & Quantum & 8 & 320& 10 & 400 \\ 
        PQWGAN & Quantum & 80 & 800 & 168 & 5120 \\ 
        WGAN-GP & Classical & - & 945152 & - & 1732352 \\
        \bottomrule
    \end{tabular}
\end{table}
% \subsection{Additional Ablations}
% \paragraph{Quantum-Prior Ablation}
% As shown in Figure~\ref{crop_ablation_study6},
% we conduct a quantum-prior ablation on the 16$\times$16 digit 0. The detailed results further support the main QFL-calibration finding reported above.

\subsection{Auxiliary-System Ablation}
BasicQGAN employs an $n$-qubit unitary generator without auxiliary qubits. To examine whether QFL calibration also applies when the quantum generator is extended with an auxiliary system, we construct BasicQGAN w/ Ancilla qubit by appending one auxiliary qubit initialized identically as $R_Y(0)\lvert 0\rangle=\lvert 0\rangle$. Thus, the $i$-th calibrated initial state becomes
\begin{equation}
\widetilde{\rho}_i
=
\rho_i\otimes\lvert 0\rangle\langle 0\rvert.
\end{equation}
By the multiplicativity of quantum fidelity under tensor products \citep{nielsen2010quantum},
\begin{equation}
\begin{aligned}
F(\widetilde{\rho}_i,\widetilde{\rho}_j)
&=F(\rho_i,\rho_j)
F(\lvert 0\rangle\langle0\rvert,\lvert 0\rangle\langle0\rvert)\\
&=F(\rho_i,\rho_j).
\end{aligned}
\end{equation}
Therefore, appending the auxiliary state leaves the prior-induced QFL unchanged, showing that the QFL-calibrated initial-state construction extends directly to the enlarged system. The subsequent shared unitary evolution also preserves the joint-state fidelities. We compare BasicQGAN with BasicQGAN w/ Ancilla qubit on MNIST at $16\times16$ resolution. As shown in Figures~\ref{Fid_3_2} and~\ref{Fid_B_VS_B_WAQ}, BasicQGAN w/ Ancilla qubit produces recognizable samples and obtains FID scores comparable to those of BasicQGAN. These results demonstrate that QFL calibration is also applicable to this auxiliary-system architecture.

\begin{figure}[H]
    \centering
   \includegraphics[width=\linewidth]{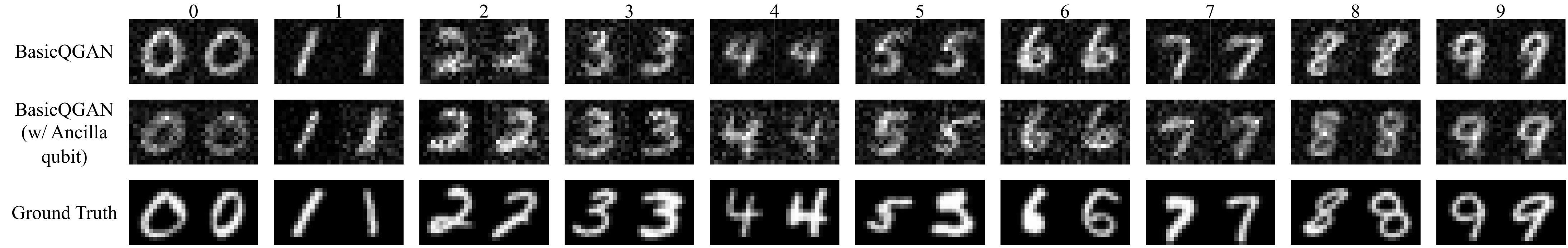}
    \caption{Comparison between BasicQGAN and BasicQGAN w/ Ancilla qubit on MNIST at $16\times16$ resolution.}
    \label{Fid_3_2}
\end{figure}

\begin{figure}[H]
    \centering
   \includegraphics[width=0.8\linewidth]{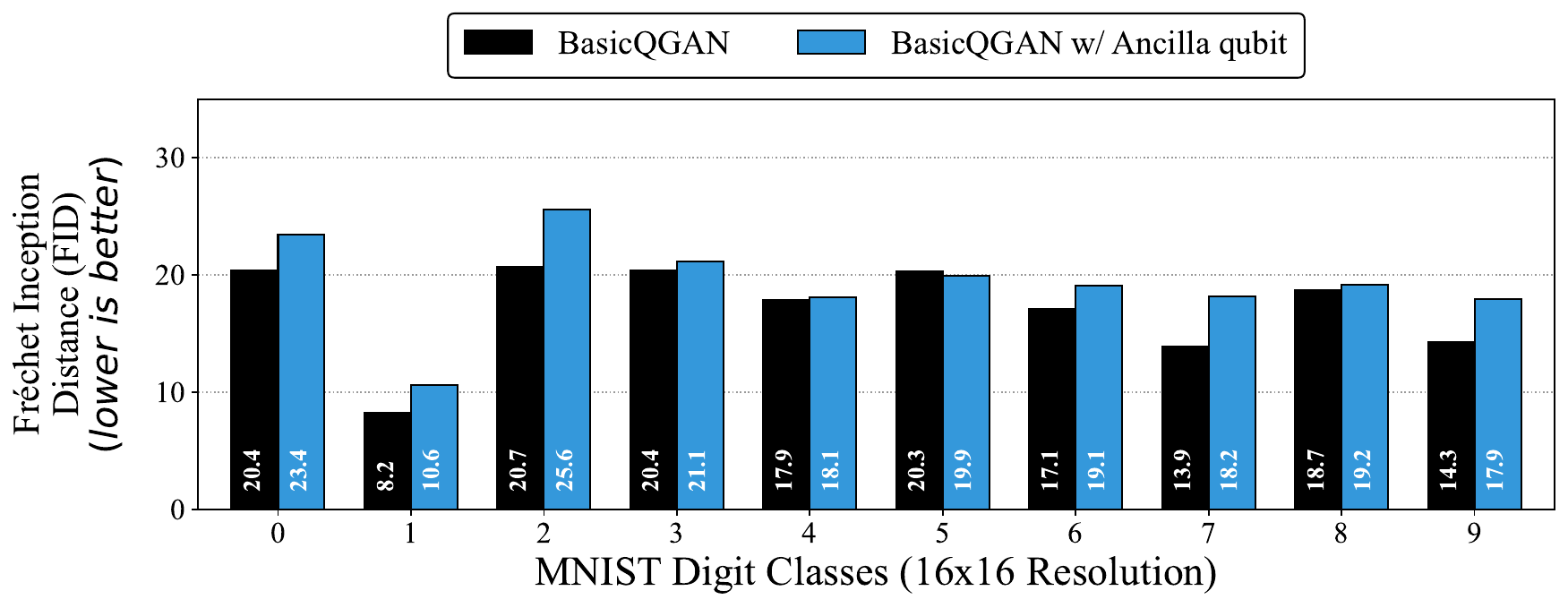}
    \caption{FID comparison between BasicQGAN and BasicQGAN w/ Ancilla qubit across different MNIST classes at $16\times16$ resolution.}
    \label{Fid_B_VS_B_WAQ}
\end{figure}

Further details regarding the aforementioned experiments are provided in ~\ref{app_6}. We also report additional experiments on the learning rate used by the initial-state ensemble preparation algorithm and on the impact of quantum circuit depth on BasicQGAN's performance in \ref{app_7} and~\ref{app_8}.

\section{Conclusion and Limitations}
\label{Conclusion}
We studied single-circuit, end-to-end pixel-level image generation with QGANs from the perspective of quantum-prior geometry. Our analysis identified the Quantum Fidelity Landscape (QFL) as a pairwise-fidelity structure fixed by the initial-state ensemble and preserved by shared unitary evolution. Under a fixed Lipschitz readout, this invariant further imposes a one-sided constraint on the separation of decoded samples, connecting the geometry of the quantum prior to the behavior of the generated distribution.

Guided by this insight, we developed BasicQGAN, which calibrates the prior-induced QFL against the data-induced QFL before WGAN-GP adversarial training. In the evaluated settings, QFL calibration avoids the training failure or mode collapse observed with mismatched priors and produces diverse samples. BasicQGAN achieves competitive generation performance on $16\times16$ MNIST and two additional grayscale datasets, while requiring fewer qubits and trainable parameters than representative patch-based quantum generators. The auxiliary-system experiment further shows that the same QFL calibration principle applies when the generator is extended with an ancilla qubit.

The present study has several limitations. The experiments are conducted in ideal, noise-free simulation, and the behavior on real NISQ hardware remains to be evaluated. The current validation also focuses mainly on small-scale grayscale images. In particular, BasicQGAN degrades at $28\times28$ resolution, indicating that improved circuit design and prior calibration strategies are needed for higher-resolution or multi-channel generation. Together, these results support QFL-guided prior calibration as a practical basis for compact, single-circuit QGAN image generation, while motivating further development toward higher-dimensional and hardware-based settings.

% \section{Acknowledgements}
% \label{Acknowledgements}
% This work was supported by the National Natural Science Foundation of China under Grant No. 62572292, in part by the China Scholarship Council
% (CSC) under Grant No. 202508310168, and by Shanghai Maritime University’s Top Innovative Talent Training Program for Graduate Students under Grant No. 2025YBR013.

% \section*{Impact Statement}

% This paper presents work whose goal is to advance the field of Quantum Machine Learning. There are many potential societal consequences of our work, none which we feel must be specifically highlighted here.

%%%%%%%%%%%%%%%%%%%%%%%%%%%%%%%%%%%%%%%%%%%%%%%%%%%%%%%%%%%%%%%%%%%%%%%%%%%%%%%
%%%%%%%%%%%%%%%%%%%%%%%%%%%%%%%%%%%%%%%%%%%%%%%%%%%%%%%%%%%%%%%%%%%%%%%%%%%%%%%
% APPENDIX
%%%%%%%%%%%%%%%%%%%%%%%%%%%%%%%%%%%%%%%%%%%%%%%%%%%%%%%%%%%%%%%%%%%%%%%%%%%%%%%
%%%%%%%%%%%%%%%%%%%%%%%%%%%%%%%%%%%%%%%%%%%%%%%%%%%%%%%%%%%%%%%%%%%%%%%%%%%%%%%
\newpage
\appendix
% ---- BEGIN INLINED THEORY PROOFS ----
\section{Proof of QFL Invariance}
\label{app:qfl_invariance}

Throughout the appendix, $\rho$ and $\sigma$ denote density operators on a finite-
dimensional Hilbert space, $\operatorname{Tr}$ denotes the matrix trace, and
$A^\dagger$ denotes the conjugate transpose of an operator $A$. Following the standard
definition in \citet{nielsen2010quantum}, we use the squared Uhlmann fidelity convention
\begin{equation}
 F(\rho,\sigma)
 =\left(\operatorname{Tr}\sqrt{\sqrt{\rho}\,\sigma\sqrt{\rho}}\right)^2.
 \label{eq:app_fidelity_definition}
\end{equation}
where $\sqrt{\rho}$ is the unique positive-semidefinite square root of $\rho$.
Unitary invariance is a standard property of quantum fidelity
\citep{nielsen2010quantum}: for any density operators
$\rho,\sigma$ and any unitary $U$,
\begin{equation}
 F(U\rho U^\dagger,U\sigma U^\dagger)=F(\rho,\sigma)
 \label{eq:app_unitary_invariance_theorem}
\end{equation}

\begin{proof}[Proof of Proposition~\ref{prop:qfl_invariance}]
Let $\mathcal{E}_0=\{\rho_i\}_{i=1}^{M}$ be an ensemble of $M$ initial states, let
$\theta$ denote the trainable circuit parameters, and define the generated ensemble
$\mathcal{E}_\theta=\{\rho_i^\theta\}_{i=1}^{M}$. For every pair of indices
$i,j\in\{1,\ldots,M\}$, the shared generator produces
\begin{equation}
 \rho_i^\theta=U_\theta\rho_iU_\theta^\dagger,
 \qquad
 \rho_j^\theta=U_\theta\rho_jU_\theta^\dagger.
\end{equation}
Applying Eq.~\eqref{eq:app_unitary_invariance_theorem} with $U=U_\theta$ gives
\begin{align}
 F(\rho_i^\theta,\rho_j^\theta)
 &=F(U_\theta\rho_iU_\theta^\dagger,
      U_\theta\rho_jU_\theta^\dagger)\\
 &=F(\rho_i,\rho_j).
 \label{eq:app_qfl_pair_invariance}
\end{align}
Equation~\eqref{eq:app_qfl_pair_invariance} holds for every pair and every $\theta$.
Therefore, every entry of the QFL matrix $K(\mathcal{E})$, defined in
Eq.~\eqref{eq:qfl_matrix}, satisfies
$K_{ij}(\mathcal{E}_\theta)=F(\rho_i^\theta,\rho_j^\theta)$ equals the corresponding
entry of $K_{ij}(\mathcal{E}_0)$. The complete QFL matrix, and consequently its empirical
off-diagonal distribution, is invariant under generator updates.

For completeness, in the pure-state setting used by BasicQGAN, write
$|\psi_i^\theta\rangle=U_\theta|\psi_i\rangle$ and
$|\psi_j^\theta\rangle=U_\theta|\psi_j\rangle$. Since a unitary satisfies
$U_\theta^\dagger U_\theta=I$, where $I$ is the identity operator,
\begin{align}
 F(\rho_i^\theta,\rho_j^\theta)
 &=|\langle\psi_i^\theta|\psi_j^\theta\rangle|^2\\
 &=|\langle\psi_i|U_\theta^\dagger U_\theta|\psi_j\rangle|^2\\
 &=|\langle\psi_i|\psi_j\rangle|^2
 =F(\rho_i,\rho_j),
\end{align}
which gives the same result directly.
\end{proof}

\section{Proof of the General Output-Separation Bound}
\label{app:readout_bound}

We first recall the standard definitions of trace norm and trace distance from quantum
information theory \citep{nielsen2010quantum}. For a matrix $A$, let
$\|A\|_1:=\operatorname{Tr}\sqrt{A^\dagger A}$ denote its trace norm. The trace distance
between density operators is
\begin{equation}
 D_{\mathrm{tr}}(\rho,\sigma)=\frac{1}{2}\|\rho-\sigma\|_1,
 \label{eq:app_trace_distance}
\end{equation}
For two probability vectors $p=(p_k)_{k=1}^{K}$ and $q=(q_k)_{k=1}^{K}$ over the same
$K$ measurement outcomes, we use the total variation distance convention of
\citet{gibbs2002choosing}. This quantity is also the classical trace distance used by
\citet[Section~9.1]{nielsen2010quantum}:
\begin{equation}
 \operatorname{TV}(p,q)=\frac{1}{2}\sum_k|p_k-q_k|.
 \label{eq:app_tv_distance}
\end{equation}

\paragraph{Result 1: contractivity under measurement.}
In the standard channel formulation of quantum measurements
\citep{nielsen2010quantum}, any fixed measurement $\mathcal{M}$ is a completely positive
trace-preserving (CPTP) quantum-to-classical map. Let $\Lambda$ denote an arbitrary CPTP
map. The data-processing inequality for trace distance states that trace distance is
contractive under $\Lambda$ \citep{nielsen2010quantum}:
\begin{equation}
 D_{\mathrm{tr}}(\Lambda(\rho),\Lambda(\sigma))
 \leq D_{\mathrm{tr}}(\rho,\sigma).
 \label{eq:app_trace_contractivity}
\end{equation}
For the measurement channel $\mathcal{M}$, let
$p_k=\operatorname{Tr}(E_k\rho)$ and $q_k=\operatorname{Tr}(E_k\sigma)$, where
$\{E_k\}_{k=1}^{K}$ is its positive-operator-valued measure, satisfying
$E_k\succeq0$ and $\sum_{k=1}^{K}E_k=I$. This is the standard positive-operator-valued
measure (POVM) representation \citep{nielsen2010quantum}. A quantum-to-classical
representation of the measurement stores the outcome probabilities in diagonal density
operators:
\begin{equation}
 \mathcal{M}(\rho)=\sum_k p_k|k\rangle\langle k|,
 \qquad
 \mathcal{M}(\sigma)=\sum_k q_k|k\rangle\langle k|.
\end{equation}
Using Eqs.~\eqref{eq:app_trace_distance} and~\eqref{eq:app_tv_distance}, their trace
distance is exactly
\begin{align}
 D_{\mathrm{tr}}(\mathcal{M}(\rho),\mathcal{M}(\sigma))
 &=\frac12\left\|
     \sum_k(p_k-q_k)|k\rangle\langle k|
   \right\|_1\\
 &=\frac12\sum_k|p_k-q_k|\\
 &=\operatorname{TV}(p,q).
 \label{eq:app_classical_trace_tv}
\end{align}
The second equality follows because the trace norm of a diagonal Hermitian matrix is the
sum of the absolute values of its eigenvalues. Combining
Eqs.~\eqref{eq:app_trace_contractivity} and~\eqref{eq:app_classical_trace_tv} gives
\begin{equation}
 \operatorname{TV}(p,q)\leq D_{\mathrm{tr}}(\rho,\sigma).
 \label{eq:app_measurement_contraction}
\end{equation}

\paragraph{Result 2: Fuchs--van de Graaf inequalities.}
Under the squared-fidelity convention in Eq.~\eqref{eq:app_fidelity_definition}, the
Fuchs--van de Graaf inequalities \citep{fuchs1999cryptographic,nielsen2010quantum} are
\begin{equation}
 1-\sqrt{F(\rho,\sigma)}
 \leq D_{\mathrm{tr}}(\rho,\sigma)
 \leq\sqrt{1-F(\rho,\sigma)}.
 \label{eq:app_fvdg}
\end{equation}
Only the upper bound is needed below.

\begin{proof}[Proof of Theorem~\ref{thm:readout_bound}]
Let $(\mathcal{X},d_{\mathcal X})$ be the metric space of classical outputs. Let
$h$ be the fixed classical post-processing map in
Assumption~\ref{ass:stable_readout}, and let $L<\infty$ be its Lipschitz constant with
respect to total variation distance.
Let
\begin{equation}
 \rho_i^\theta=U_\theta\rho_iU_\theta^\dagger,
 \qquad
 p_i^\theta=\mathcal{M}(\rho_i^\theta),
\end{equation}
and define $\rho_j^\theta$ and $p_j^\theta$ analogously. Because
$\mathcal{D}=h\circ\mathcal{M}$,
\begin{equation}
 \mathcal{D}(\rho_i^\theta)=h(p_i^\theta),
 \qquad
 \mathcal{D}(\rho_j^\theta)=h(p_j^\theta).
 \label{eq:app_readout_composition}
\end{equation}

We now apply the assumptions and standard results in sequence. First, the stable-readout
condition in Assumption~\ref{ass:stable_readout} gives
\begin{align}
 &d_{\mathcal X}\!\left(
   \mathcal{D}(\rho_i^\theta),\mathcal{D}(\rho_j^\theta)
   \right)\nonumber\\
 &\quad=d_{\mathcal X}\!\left(h(p_i^\theta),h(p_j^\theta)\right)\\
 &\quad\leq L\operatorname{TV}(p_i^\theta,p_j^\theta).
 \label{eq:app_step_lipschitz}
\end{align}
Second, measurement contractivity in Eq.~\eqref{eq:app_measurement_contraction} yields
\begin{equation}
 \operatorname{TV}(p_i^\theta,p_j^\theta)
 \leq D_{\mathrm{tr}}(\rho_i^\theta,\rho_j^\theta).
 \label{eq:app_step_measurement}
\end{equation}
Third, the upper Fuchs--van de Graaf bound in Eq.~\eqref{eq:app_fvdg} gives
\begin{equation}
 D_{\mathrm{tr}}(\rho_i^\theta,\rho_j^\theta)
 \leq\sqrt{1-F(\rho_i^\theta,\rho_j^\theta)}.
 \label{eq:app_step_fvdg}
\end{equation}
Finally, Proposition~\ref{prop:qfl_invariance} gives
\begin{equation}
 F(\rho_i^\theta,\rho_j^\theta)=F(\rho_i,\rho_j).
 \label{eq:app_step_unitary}
\end{equation}
Substituting Eqs.~\eqref{eq:app_step_measurement}--\eqref{eq:app_step_unitary} into
Eq.~\eqref{eq:app_step_lipschitz}, we obtain
\begin{align}
 &d_{\mathcal X}\!\left(
   \mathcal{D}(\rho_i^\theta),\mathcal{D}(\rho_j^\theta)
   \right)\nonumber\\
 &\quad\leq L\operatorname{TV}(p_i^\theta,p_j^\theta)\\
 &\quad\leq L D_{\mathrm{tr}}(\rho_i^\theta,\rho_j^\theta)\\
 &\quad\leq L\sqrt{1-F(\rho_i^\theta,\rho_j^\theta)}\\
 &\quad= L\sqrt{1-F(\rho_i,\rho_j)},
\end{align}
which is Eq.~\eqref{eq:general_output_bound}.
\end{proof}

\paragraph{Uniformity over trainable classical decoders.}
Theorem~\ref{thm:readout_bound} also clarifies the role of a trainable classical
post-processing map. Consider a family of decoders
$\mathcal{D}_\phi=h_\phi\circ\mathcal{M}$ indexed by trainable parameters
$\phi\in\Phi$. For each fixed $\phi$, the proof above remains valid with the corresponding
Lipschitz constant $L_\phi:=\operatorname{Lip}(h_\phi)$. A bound that is uniform over
training requires a finite constant
\begin{equation}
 L_*:=\sup_{\phi\in\Phi}L_\phi<\infty,
 \label{eq:app_uniform_decoder_lipschitz}
\end{equation}
in which case Eq.~\eqref{eq:general_output_bound} holds with $L_*$ at every training
step. If Eq.~\eqref{eq:app_uniform_decoder_lipschitz} is not enforced, the learned map
may amplify an arbitrarily small nonzero difference between measurement vectors, so QFL
alone provides no parameter-independent bound on the final Euclidean separation. It
still cannot recover information already removed by the measurement: if
$\mathcal{M}(\rho_i^\theta)=\mathcal{M}(\rho_j^\theta)$, then every shared deterministic
decoder satisfies
$\mathcal{D}_\phi(\rho_i^\theta)=\mathcal{D}_\phi(\rho_j^\theta)$.

\begin{proof}[Proof of Eq.~\eqref{eq:output_dispersion_bound}]
For each state $\rho_i^\theta$, define its decoded output as
$x_i:=\mathcal{D}(\rho_i^\theta)\in\mathbb{R}^{d_{\mathrm{out}}}$, where
$d_{\mathrm{out}}$ is the number of output coordinates. Let
$\bar{x}=M^{-1}\sum_{i=1}^{M}x_i$. We first derive the pairwise-variance identity rather
than invoking it without proof. Expanding the squared Euclidean distance gives
\begin{align}
 \sum_{i,j=1}^{M}\|x_i-x_j\|_2^2
 &=\sum_{i,j=1}^{M}
   \left(\|x_i\|_2^2+\|x_j\|_2^2-2x_i^\top x_j\right)\\
 &=2M\sum_{i=1}^{M}\|x_i\|_2^2
   -2\left\|\sum_{i=1}^{M}x_i\right\|_2^2.
 \label{eq:app_pairwise_expansion}
\end{align}
The second equality uses
$\sum_{i,j}\|x_i\|_2^2=M\sum_i\|x_i\|_2^2$ and
$\sum_{i,j}x_i^\top x_j=\|\sum_i x_i\|_2^2$. Since
$\sum_i x_i=M\bar{x}$, Eq.~\eqref{eq:app_pairwise_expansion} becomes
\begin{equation}
 \sum_{i,j=1}^{M}\|x_i-x_j\|_2^2
 =2M\sum_{i=1}^{M}\|x_i\|_2^2-2M^2\|\bar{x}\|_2^2.
 \label{eq:app_pairwise_expansion_mean}
\end{equation}
On the other hand, direct expansion around the mean gives
\begin{align}
 \sum_{i=1}^{M}\|x_i-\bar{x}\|_2^2
 &=\sum_{i=1}^{M}
   \left(\|x_i\|_2^2-2x_i^\top\bar{x}+\|\bar{x}\|_2^2\right)\\
 &=\sum_{i=1}^{M}\|x_i\|_2^2-M\|\bar{x}\|_2^2,
 \label{eq:app_centered_expansion}
\end{align}
where the last equality uses $\sum_i x_i=M\bar{x}$. Comparing
Eqs.~\eqref{eq:app_pairwise_expansion_mean} and~\eqref{eq:app_centered_expansion} yields
\begin{equation}
 \frac1M\sum_{i=1}^{M}\|x_i-\bar{x}\|_2^2
 =\frac{1}{2M^2}\sum_{i,j=1}^{M}\|x_i-x_j\|_2^2.
 \label{eq:app_pairwise_variance_identity}
\end{equation}

For the Euclidean output metric, Theorem~\ref{thm:readout_bound} states that
\begin{equation}
 \|x_i-x_j\|_2
 \leq L\sqrt{1-F(\rho_i,\rho_j)}.
 \label{eq:app_pairwise_output_bound}
\end{equation}
Both sides of Eq.~\eqref{eq:app_pairwise_output_bound} are nonnegative. Squaring therefore
preserves the inequality and gives
\begin{equation}
 \|x_i-x_j\|_2^2
 \leq L^2\left(1-F(\rho_i,\rho_j)\right).
 \label{eq:app_squared_pairwise_output_bound}
\end{equation}
Summing Eq.~\eqref{eq:app_squared_pairwise_output_bound} over all ordered pairs $(i,j)$,
dividing by $2M^2$, and using Eq.~\eqref{eq:app_pairwise_variance_identity}, we obtain
\begin{align}
 \frac1M\sum_{i=1}^{M}\|x_i-\bar{x}\|_2^2
 &=\frac{1}{2M^2}\sum_{i,j=1}^{M}\|x_i-x_j\|_2^2\\
 &\leq\frac{L^2}{2M^2}\sum_{i,j=1}^{M}
 \left(1-F(\rho_i,\rho_j)\right),
\end{align}
which proves Eq.~\eqref{eq:output_dispersion_bound}.
\end{proof}

\section{Proof of the Magnitude-Readout Bound}
\label{app:magnitude_bound}

\begin{proof}[Proof of Corollary~\ref{cor:magnitude_bound}]
Let $\{|k\rangle\}_{k=1}^{d}$ be an orthonormal computational basis of a
$d$-dimensional Hilbert space. For $i,j$ indexing two pure states, let
\begin{equation}
 |\psi_i\rangle=\sum_{k=1}^{d}c_{ik}|k\rangle,
 \qquad
 |\psi_j\rangle=\sum_{k=1}^{d}c_{jk}|k\rangle,
\end{equation}
where $c_{ik},c_{jk}\in\mathbb{C}$ are complex amplitudes, and let
$\rho_i:=|\psi_i\rangle\langle\psi_i|$ and
$\rho_j:=|\psi_j\rangle\langle\psi_j|$ be the corresponding density operators.
Define their nonnegative magnitude vectors
\begin{equation}
 a_i=(|c_{i1}|,\ldots,|c_{id}|)^\top,
 \qquad
 a_j=(|c_{j1}|,\ldots,|c_{jd}|)^\top.
\end{equation}
Because quantum states are normalized,
\begin{equation}
 \|a_i\|_2^2=\sum_k|c_{ik}|^2=1,
 \qquad
 \|a_j\|_2^2=\sum_k|c_{jk}|^2=1.
 \label{eq:app_magnitude_normalization}
\end{equation}
For pure states, the standard reduction of Uhlmann fidelity
\citep{nielsen2010quantum}, under our squared-fidelity convention, is
\begin{equation}
 \sqrt{F(\rho_i,\rho_j)}=|\langle\psi_i|\psi_j\rangle|.
 \label{eq:app_pure_fidelity_root}
\end{equation}
Here $(\cdot)^*$ denotes complex conjugation and $(\cdot)^\top$ denotes vector
transpose. Expanding the inner product and applying the triangle inequality for complex numbers,
$|\sum_k z_k|\leq\sum_k|z_k|$, gives
\begin{align}
 \sqrt{F(\rho_i,\rho_j)}
 &=\left|\sum_{k=1}^{d}c_{ik}^*c_{jk}\right|\\
 &\leq\sum_{k=1}^{d}|c_{ik}^*c_{jk}|\\
 &=\sum_{k=1}^{d}|c_{ik}||c_{jk}|\\
 &=a_i^\top a_j.
 \label{eq:app_fidelity_magnitude_inner_product}
\end{align}
The squared Euclidean distance between the two magnitude vectors is
\begin{align}
 \|a_i-a_j\|_2^2
 &=\|a_i\|_2^2+\|a_j\|_2^2-2a_i^\top a_j\\
 &=2-2a_i^\top a_j,
 \label{eq:app_magnitude_distance_expansion}
\end{align}
where Eq.~\eqref{eq:app_magnitude_normalization} is used in the second line. From
Eq.~\eqref{eq:app_fidelity_magnitude_inner_product},
$a_i^\top a_j\geq\sqrt{F(\rho_i,\rho_j)}$; multiplying by $-2$ reverses the inequality:
\begin{equation}
 -2a_i^\top a_j\leq-2\sqrt{F(\rho_i,\rho_j)}.
\end{equation}
Adding $2$ to both sides and using Eq.~\eqref{eq:app_magnitude_distance_expansion} gives
\begin{equation}
 \|a_i-a_j\|_2^2
 \leq2\left(1-\sqrt{F(\rho_i,\rho_j)}\right),
\end{equation}
which proves Eq.~\eqref{eq:magnitude_distance_bound}. The derivation uses amplitude
magnitudes only and therefore does not require $c_{ik}$ or $c_{jk}$ to be real.
\end{proof}

\paragraph{Connection to the implemented readout.}
In our implementation, the circuit ansatz consists only of $R_Y$ and CNOT gates.
Because both gates have real-valued matrix representations in the computational basis
and the initial state $|0\rangle^{\otimes N}$ is real, the ideal output statevector has
real-valued amplitudes. Therefore, the practical readout
$|\operatorname{Re}(c_{ik})|$ coincides with the amplitude magnitude $|c_{ik}|$,
up to floating-point error, and the bound in
Eq.~\eqref{eq:magnitude_distance_bound} applies directly. For an extension containing
complex-valued gates, the corresponding implementation should instead use
$|c_{ik}|=\sqrt{p_{ik}}$, where $p_{ik}=|c_{ik}|^2$ is the computational-basis
measurement probability.

\paragraph{Effect of the BasicQGAN normalization.}
For a vector $a=(a_k)_{k=1}^{d}$, define
$\|a\|_\infty:=\max_{1\leq k\leq d}|a_k|$, and let
$\mathbf{1}\in\mathbb{R}^{d}$ denote the all-ones vector. Ignoring the numerical
stabilizer $\epsilon$ for the displayed algebraic identity, the decoder
is
\begin{equation}
 \widetilde{x}_i=2\frac{a_i}{\|a_i\|_\infty}-\mathbf{1}.
\end{equation}
Therefore,
\begin{equation}
 \widetilde{x}_i+\mathbf{1}=2\frac{a_i}{\|a_i\|_\infty}.
 \label{eq:app_shifted_decoder}
\end{equation}
Taking the Euclidean norm of Eq.~\eqref{eq:app_shifted_decoder} and using
$\|a_i\|_2=1$ gives
\begin{equation}
 \|\widetilde{x}_i+\mathbf{1}\|_2
 =\frac{2\|a_i\|_2}{\|a_i\|_\infty}
 =\frac{2}{\|a_i\|_\infty}.
\end{equation}
Dividing Eq.~\eqref{eq:app_shifted_decoder} by this norm yields
\begin{equation}
 \frac{\widetilde{x}_i+\mathbf{1}}
      {\|\widetilde{x}_i+\mathbf{1}\|_2}=a_i.
\end{equation}
Thus, sample-wise maximum normalization changes the radial scale but preserves the
normalized magnitude direction. The implementation includes a fixed numerical
stabilizer $\epsilon>0$ in both normalization steps. Writing
$s_i:=\|a_i\|_1$ and $u_i:=a_i/(s_i+\epsilon)$, the implemented output is
\begin{equation}
 T_\epsilon(a_i)+\mathbf{1}
 =2\frac{u_i}{\|u_i\|_\infty+\epsilon}
 =\frac{2a_i}{\|a_i\|_\infty+\epsilon(s_i+\epsilon)}.
 \label{eq:app_implemented_decoder_epsilon}
\end{equation}
The right-hand side is a positive sample-dependent scalar multiple of $a_i$; hence
$(T_\epsilon(a_i)+\mathbf{1})/
\|T_\epsilon(a_i)+\mathbf{1}\|_2=a_i$ still holds because $\|a_i\|_2=1$.
Moreover, $\|a_i\|_\infty\geq1/\sqrt d$ on the normalized amplitude domain, so the
denominator in Eq.~\eqref{eq:app_implemented_decoder_epsilon} is bounded away from zero.
Because the absolute-value, $\ell_1$-norm, maximum, and quotient operations are
Lipschitz on this restricted domain, $T_\epsilon$ is also Lipschitz there. The
stabilizer changes the corresponding constant but introduces no trainable mapping.

The preceding identity establishes recoverability of $a_i$, but recoverability should
not be confused with preservation of Euclidean distance. We next quantify the metric
distortion for the idealized decoder without $\epsilon$. Let
$T(a):=2a/\|a\|_\infty-\mathbf{1}$ and consider arbitrary
$a,b\in\mathbb{R}_+^d$ satisfying $\|a\|_2=\|b\|_2=1$. Write
$\alpha:=\|a\|_\infty$ and $\beta:=\|b\|_\infty$. Since a vector's largest coordinate
is at least its root-mean-square coordinate,
\begin{equation}
 \alpha\geq\frac{\|a\|_2}{\sqrt d}=\frac1{\sqrt d},
 \qquad
 \beta\geq\frac{\|b\|_2}{\sqrt d}=\frac1{\sqrt d}.
 \label{eq:app_infinity_lower_bound}
\end{equation}
Using the triangle inequality, $\|b\|_2=1$, and the reverse triangle inequality
$|\alpha-\beta|\leq\|a-b\|_\infty\leq\|a-b\|_2$, we obtain
\begin{align}
 \|T(a)-T(b)\|_2
 &=2\left\|\frac{a}{\alpha}-\frac{b}{\beta}\right\|_2\\
 &\leq 2\left(
     \frac{\|a-b\|_2}{\alpha}
     +\|b\|_2\left|\frac1\alpha-\frac1\beta\right|
   \right)\\
 &=2\left(
     \frac{\|a-b\|_2}{\alpha}
     +\frac{|\alpha-\beta|}{\alpha\beta}
   \right)\\
 &\leq 2(\sqrt d+d)\|a-b\|_2.
 \label{eq:app_decoder_upper_lipschitz}
\end{align}
This proves the upper inequality in Eq.~\eqref{eq:decoder_metric_distortion}.

For the reverse direction, set $u:=T(a)+\mathbf{1}=2a/\alpha$ and
$v:=T(b)+\mathbf{1}=2b/\beta$. The standard normalization inequality
\begin{equation}
 \left\|\frac{u}{\|u\|_2}-\frac{v}{\|v\|_2}\right\|_2
 \leq \frac{2\|u-v\|_2}{\min\{\|u\|_2,\|v\|_2\}}
 \label{eq:app_normalization_inequality}
\end{equation}
follows by adding and subtracting $v/\|u\|_2$, applying the triangle inequality,
and using $|\|u\|_2-\|v\|_2|\leq\|u-v\|_2$. Here
$\|u\|_2=2/\alpha\geq2$ and $\|v\|_2=2/\beta\geq2$. Moreover,
$u/\|u\|_2=a$ and $v/\|v\|_2=b$. Equation~\eqref{eq:app_normalization_inequality}
therefore gives
\begin{equation}
 \|a-b\|_2\leq\|u-v\|_2=\|T(a)-T(b)\|_2,
 \label{eq:app_decoder_lower_lipschitz}
\end{equation}
which proves the lower inequality in Eq.~\eqref{eq:decoder_metric_distortion}.
Finally, substituting Eq.~\eqref{eq:magnitude_distance_bound} into
Eq.~\eqref{eq:app_decoder_upper_lipschitz} gives
Eq.~\eqref{eq:decoded_fidelity_bound}. Hence maximum normalization does not preserve
Euclidean distance exactly, but it is a fixed bi-Lipschitz transformation on the
nonnegative unit sphere. The numerical stabilizer changes the constants slightly while
preserving the same sample direction; no trainable transformation is introduced.

\section{Consistency of Finite-Ensemble QFL Statistics}
\label{app:qfl_stabilization}

We use the strong law of large numbers for U-statistics. Let
$X_1,X_2,\ldots$ be independent and identically distributed (i.i.d.) random elements,
let $r\geq1$ be the order of the statistic, and let
$h(X_1,\ldots,X_r)$ be a symmetric integrable kernel, where integrability means
$\mathbb{E}[|h(X_1,\ldots,X_r)|]<\infty$. Then the associated U-statistic
converges almost surely to $\mathbb{E}[h(X_1,\ldots,X_r)]$, where
$\mathbb{E}$ denotes expectation \citep[Chapter~5]{serfling1980approximation}.

Let $\rho_1,\ldots,\rho_M$ be i.i.d. random states drawn from a fixed population
distribution over density operators, and define the pairwise fidelity
$G_{ij}:=F(\rho_i,\rho_j)\in[0,1]$. Although the collection
$\{G_{ij}:i<j\}$ is dependent, averages over all unordered pairs have the U-statistic
structure needed below.

\begin{theorem}[Consistency of finite-ensemble QFL statistics]
\label{thm:app_qfl_statistic_consistency}
For any bounded measurable function $\varphi:[0,1]\to\mathbb{R}$, define
\begin{equation}
 U_M(\varphi):=\frac{2}{M(M-1)}
 \sum_{1\leq i<j\leq M}\varphi(G_{ij}).
 \label{eq:app_general_qfl_ustat}
\end{equation}
Then
\begin{equation}
 U_M(\varphi)\xrightarrow{\mathrm{a.s.}}
 \mathbb{E}[\varphi(F(\rho_1,\rho_2))]
 \qquad\text{as }M\to\infty.
 \label{eq:app_general_qfl_convergence}
\end{equation}
Consequently, as $M$ increases, the empirical QFL histogram converges almost surely to
the corresponding population histogram when a common finite bin partition is used
across ensemble sizes. Separately, for a fixed bounded KDE kernel, fixed
bandwidth, and fixed finite evaluation grid, the KDE values used in the calibration
loss are also consistent. The empirical mean, variance, and composite loss in
Eq.~\eqref{eq:qfl_loss} therefore converge almost surely to their corresponding
population quantities for every fixed prior parameter $\alpha$.
\end{theorem}

\begin{proof}
Define the two-argument kernel
\begin{equation}
 h_\varphi(\rho,\sigma)
 :=\varphi(F(\rho,\sigma)).
 \label{eq:app_qfl_test_kernel}
\end{equation}
Quantum fidelity is symmetric, so $h_\varphi(\rho,\sigma)=h_\varphi(\sigma,\rho)$.
Because $\varphi$ is bounded, there exists $C_\varphi<\infty$ such that
$|h_\varphi(\rho,\sigma)|\leq C_\varphi$, and hence
$\mathbb{E}|h_\varphi(\rho_1,\rho_2)|<\infty$. Equation~\eqref{eq:app_general_qfl_ustat}
is therefore an order-two U-statistic with a symmetric integrable kernel. Applying the
strong law for U-statistics stated above proves
Eq.~\eqref{eq:app_general_qfl_convergence}.

We now derive the statistics used in the paper as special cases. First, for a measurable
histogram bin $A\subset[0,1]$, choose
$\varphi_A(t):=\mathbf{1}\{t\in A\}$, where $\mathbf{1}\{E\}$ is one if event $E$
occurs and zero otherwise. Then $U_M(\varphi_A)$ is the empirical QFL mass in $A$, and
Eq.~\eqref{eq:app_general_qfl_convergence} gives
\begin{equation}
 U_M(\varphi_A)\xrightarrow{\mathrm{a.s.}}
 \Pr(F(\rho_1,\rho_2)\in A).
 \label{eq:app_histogram_bin_convergence}
\end{equation}
For a fixed finite histogram partition, this convergence holds simultaneously for all
bins because there are only finitely many of them. Hence the vector of empirical QFL bin
masses converges almost surely to the corresponding population bin-mass vector, which
formalizes convergence of the empirical QFL histogram to the corresponding population
histogram under a common bin partition.

Second, let $K:\mathbb{R}\to\mathbb{R}$ be the bounded kernel used by the KDE, let
$h>0$ be its fixed bandwidth, and define, for a fixed grid point $q_b$,
\begin{equation}
 \varphi_{q_b,h}(t):=\frac{1}{h}K\!\left(\frac{q_b-t}{h}\right).
 \label{eq:app_kde_test_function}
\end{equation}
The empirical KDE at $q_b$ is exactly
\begin{equation}
 \widehat f_M(q_b)=U_M(\varphi_{q_b,h}).
 \label{eq:app_kde_as_ustat}
\end{equation}
Since $K$ is bounded and $h$ is fixed, $\varphi_{q_b,h}$ is bounded. Therefore,
\begin{equation}
 \widehat f_M(q_b)\xrightarrow{\mathrm{a.s.}}
 \mathbb{E}\!\left[
 \frac1h K\!\left(\frac{q_b-F(\rho_1,\rho_2)}{h}\right)
 \right].
 \label{eq:app_kde_grid_convergence}
\end{equation}
Because the grid $\{q_b\}_{b=1}^{B}$ is finite, the convergence is simultaneous for all
$B$ grid values. Thus the KDE values used by the finite-grid calibration objective are
consistent estimators of their population counterparts.

Third, choosing $\varphi_1(t)=t$ and $\varphi_2(t)=t^2$ is valid because
$t,t^2\in[0,1]$. Thus the empirical first and second moments satisfy
\begin{align}
 \widehat m_M:=U_M(\varphi_1)
 &\xrightarrow{\mathrm{a.s.}}m:=\mathbb{E}[F(\rho_1,\rho_2)],\\
 \widehat s_M:=U_M(\varphi_2)
 &\xrightarrow{\mathrm{a.s.}}s:=\mathbb{E}[F(\rho_1,\rho_2)^2].
 \label{eq:app_moment_convergence}
\end{align}
For the empirical variance convention used in Eq.~\eqref{eq:qfl_loss},
$\widehat v_M=\widehat s_M-\widehat m_M^2$. Continuity of subtraction and squaring then
gives
\begin{equation}
 \widehat v_M\xrightarrow{\mathrm{a.s.}}s-m^2
 =\operatorname{Var}(F(\rho_1,\rho_2)).
 \label{eq:app_variance_convergence}
\end{equation}

Apply these results separately to the prior and data ensembles. Finite sums, absolute
values, addition, and subtraction are continuous operations. Hence the discretized KDE
discrepancy in Eq.~\eqref{eq:kde_discrepancy}, the two moment discrepancies, and their
sum $\mathcal{L}_{\mathrm{QFL}}$ converge almost surely to the same expressions formed
from the corresponding population smoothed KDE values and moments.
\end{proof}

\paragraph{Implication for QFL stability}
Consequently, as $M$ increases, the empirical QFL histogram converges to the
corresponding population histogram when a common bin partition is used across ensemble
sizes. This is the stability property illustrated by the finite-ensemble comparisons in
Figure~\ref{statistic}.

% ---- END INLINED THEORY PROOFS ----

\section{Additional Quantum Fidelity Landscape (QFL) Stability Validation}
\label{app_2}
% To further demonstrate the stability of QFL, we extended our analysis beyond the digit '0' presented in the main text. Fig.\ref{crop_Appendix_Static_Num1} and Fig.\ref{crop_Appendix_Static_Num2} 
% provide visualizations of the QFLs derived from the digit '1' class and the digit '2' class of the 16x16 MNIST dataset, respectively.

To further demonstrate the stability of QFL, we extend the analysis beyond digit~0 presented in the main text. Figure~\ref{crop_Appendix_Static_Num1} and Figure~\ref{crop_Appendix_Static_Num2} visualize the QFLs for digits~1 and~2 of the $16\times16$ MNIST dataset, respectively.

\begin{figure}[H]
    \centering
   \includegraphics[width=1\textwidth, trim=0cm 0cm 0cm 0cm, clip]{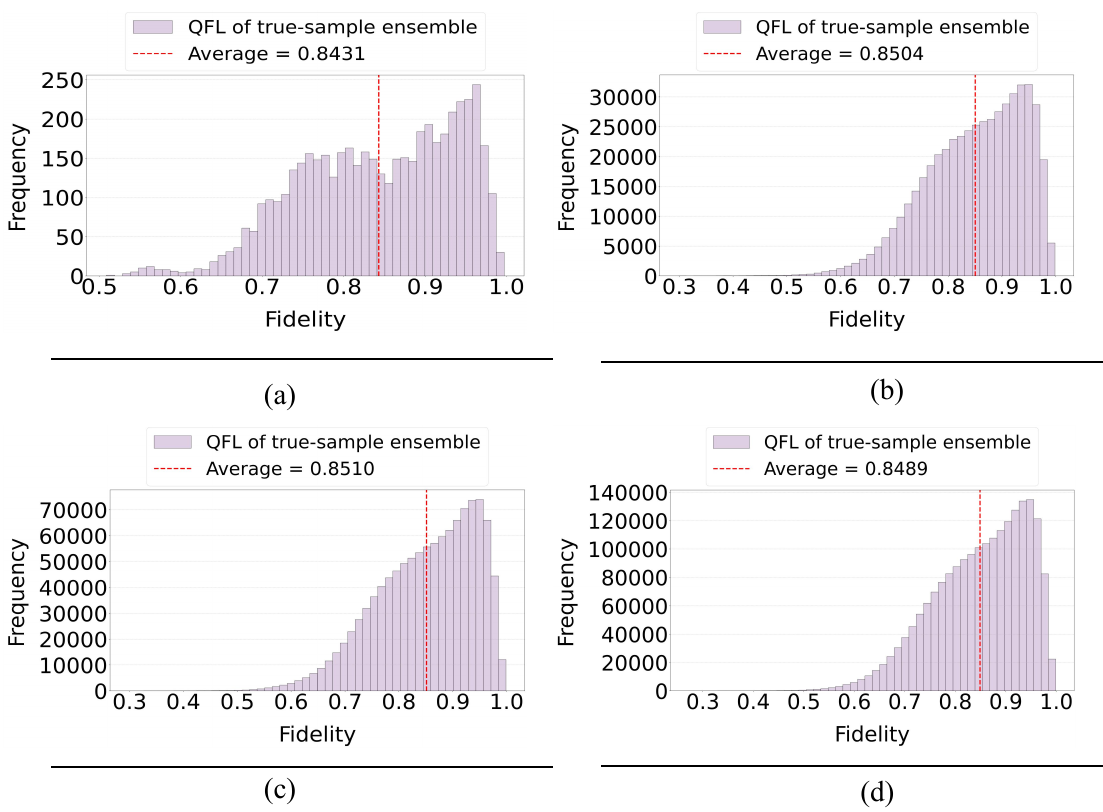}
%     \caption{Calculating the QFL for different numbers of true-sample states for the digit '1' class of the 16x16 MNIST dataset. (a)-(d) correspond to sample sizes of 500,
% 1000, 1500, and 2000, respectively.}
    \caption{QFLs computed from different numbers of real samples for digit~1 in the $16\times16$ MNIST dataset. Subfigures (a)--(d) use sample sizes of 500, 1000, 1500, and 2000, respectively.}
    \label{crop_Appendix_Static_Num1}
\end{figure}

\begin{figure}[H]
    \centering
   \includegraphics[width=1\textwidth, trim=0cm 0cm 0cm 0cm, clip]{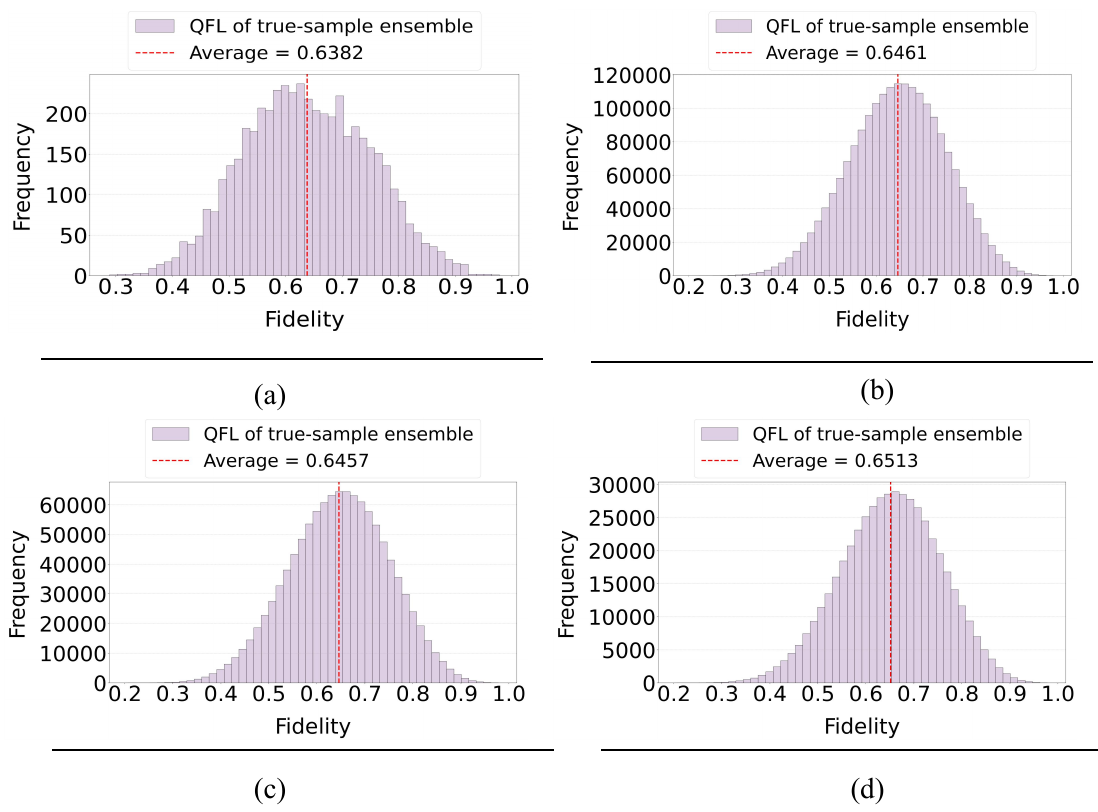}
    % \caption{Calculating the QFL for different numbers of true-sample states for the digit '2' class of the 16x16 MNIST dataset. (a)-(d) correspond to sample sizes of 500, 1000, 1500, and 2000, respectively.}
        \caption{QFLs computed from different numbers of real samples for digit~2 in the $16\times16$ MNIST dataset. Subfigures (a)--(d) use sample sizes of 500, 1000, 1500, and 2000, respectively.}
    \label{crop_Appendix_Static_Num2}
\end{figure}

\section{Achieving QFL Alignment on Diverse Datasets using the Quantum Fidelity Landscape (QFL) optimization algorithm}
\label{app_3}
This appendix provides additional QFL-alignment results achieved by our initial-state ensemble preparation algorithm. As highlighted in the main paper, accurate alignment between the initial-state ensemble and the true-sample ensemble is crucial for BasicQGAN. Figures~\ref{crop_Appendix_QFL_size16_0to9} and ~\ref{crop_Appendix_QFL_size28_0to9} present the QFL-alignment results for each digit class (0--9) on the $16\times16$ and $28\times28$ MNIST datasets, respectively.

\begin{figure}[H]
    \centering
    \includegraphics[width=0.8\textwidth, trim=0cm 0cm 0cm 0cm, clip]{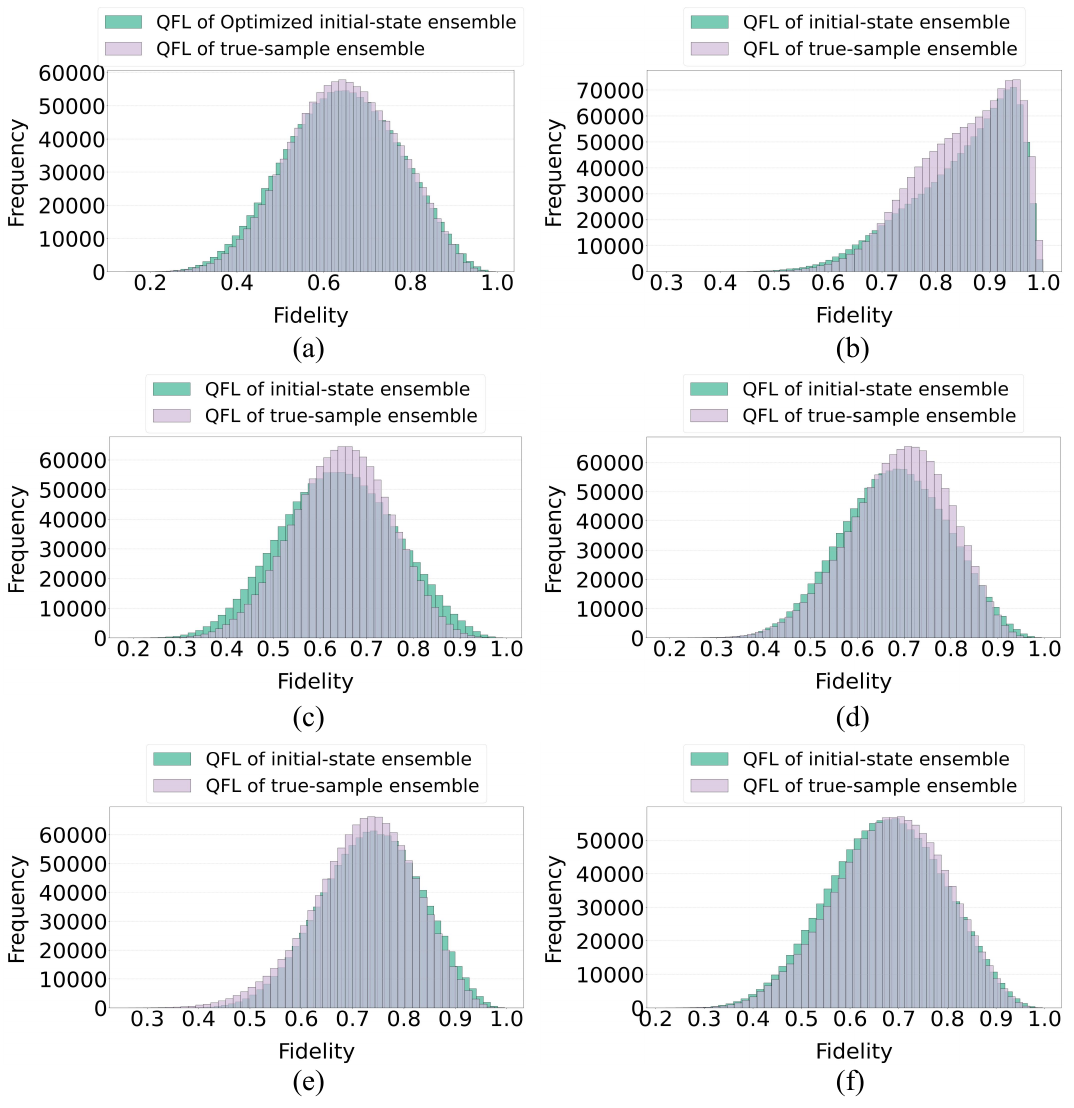}
    \includegraphics[width=0.8\textwidth, trim=0cm 0cm 0cm 0cm, clip]{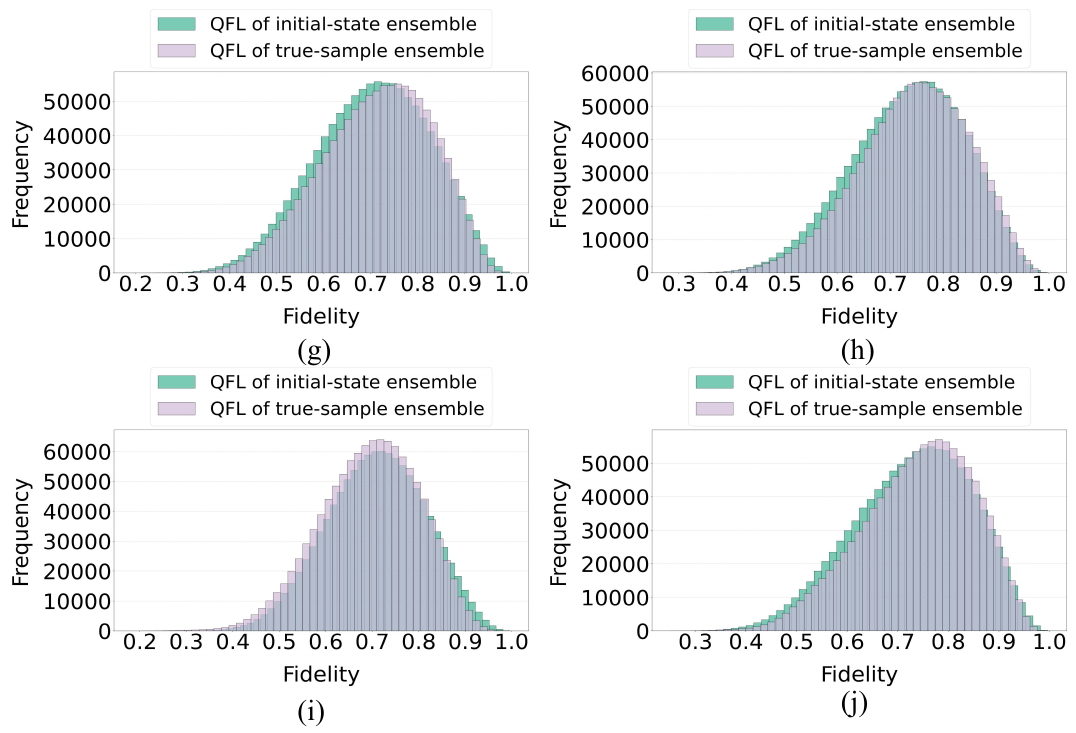}
    \caption{QFL alignment on $16\times16$ MNIST (digits 0--9). Subfigures (a)--(j) show, for each digit, the QFL of the optimized initial-state ensemble and the true-sample ensemble.}
    \label{crop_Appendix_QFL_size16_0to9}
\end{figure}
\clearpage

\begin{figure}[p]
    \centering
    \includegraphics[width=0.8\textwidth, trim=0cm 0cm 0cm 0cm, clip]{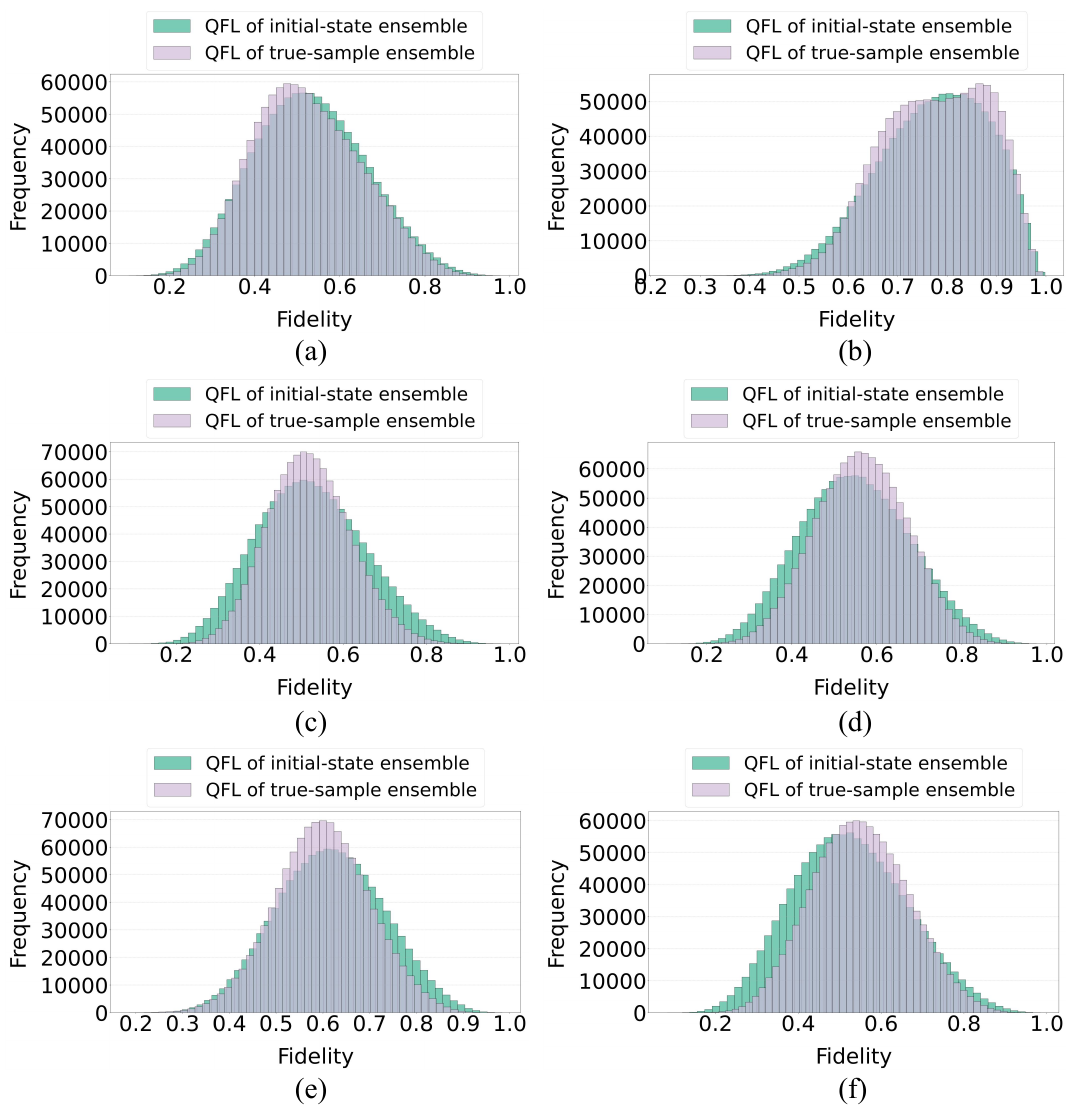}
    \includegraphics[width=0.8\textwidth, trim=0cm 0cm 0cm 0cm, clip]{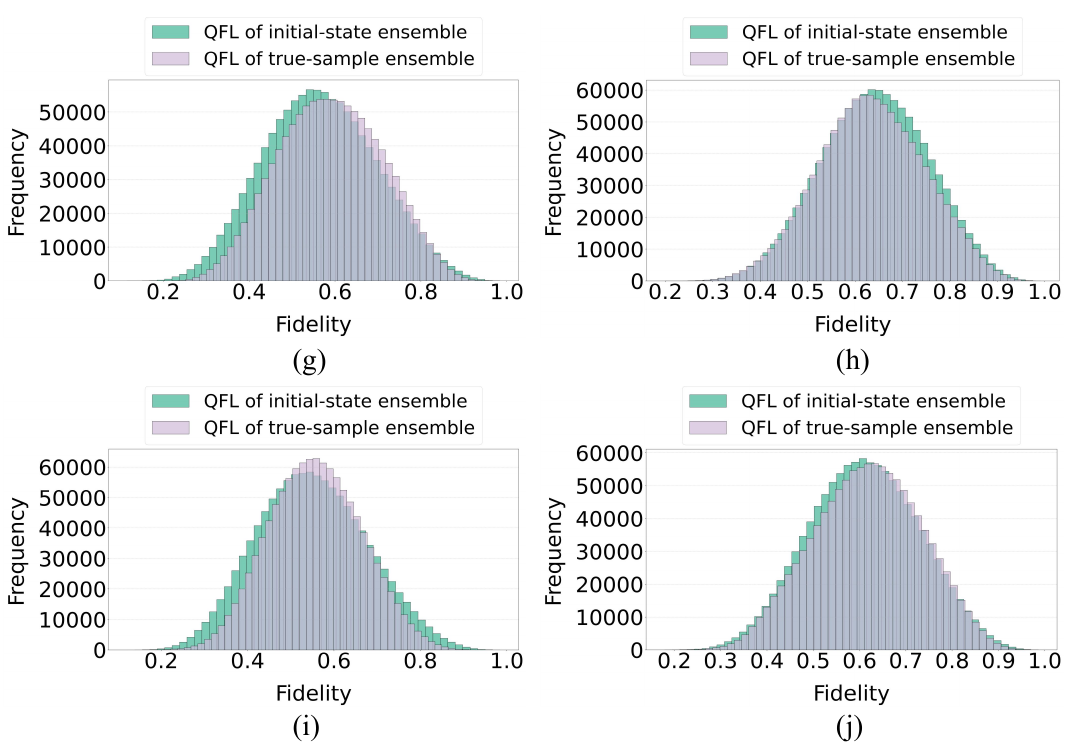}
    \caption{QFL alignment on $28\times28$ MNIST (digits 0--9). Subfigures (a)--(j) show, for each digit, the QFL of the optimized initial-state ensemble and the true-sample ensemble.}
    \label{crop_Appendix_QFL_size28_0to9}
\end{figure}

\clearpage
\section{The Objective Function and Training Algorithm of BasicQGAN}
\label{app_4}

\textbf{Objective Function} Here, we detail the specific objective functions utilized by the generator and discriminator within the BasicQGAN framework. Specifically, the discriminator aims to distinguish between real and generated data.
The corresponding loss function is given by: 
% \ref{Loss:D}.
% \begin{equation}
% L_D ={}& {E}_{x\sim p_{\text{data}}(x)} [D({x})] 
%       - {E}_{Z'\sim p(Z')}[D(G({Z'}))]
%      & - \lambda {E}_{\hat{{x}} \sim P_{\text{data}}(\hat{{x}})} \left[\left(\left\|\nabla_{\hat{{x}}} D(\hat{{x}})\right\|_2-1\right)^2\right]
% \label{Loss:D}
% \end{equation}

\begin{equation}
\label{Loss:D}
L_D
= \mathbb{E}_{x\sim p_{\text{data}}}\!\left[D(x)\right]
- \mathbb{E}_{z'\sim p(z')}\!\left[D\!\left(G(z')\right)\right]
- \lambda\, \mathbb{E}_{\hat x\sim p_{\hat x}}\!\left[\left(\left\|\nabla_{\hat x}D(\hat x)\right\|_2-1\right)^2\right].
\end{equation}

where $\hat{x}$ is uniformly interpolated between a real sample $x \sim p_{\text{data}}$ and a generated sample $\tilde{x} \sim p_G$, $p(Z')$ denotes the \emph{optimized distribution over the state-preparation rotation angles $Z'$}, and $\lambda$ is the gradient-penalty coefficient.
The quantum generator strives to create images so realistic that they can fool the discriminator.
% by minimizing the objective function.
The corresponding loss function is given by:
\begin{equation}
L_G = - {E}_{Z'\sim p(Z')}[D(G(Z'))]
\label{Loss:G}
\end{equation}

\textbf{Training Algorithm}
The training algorithm for BasicQGAN follows the WGAN-GP training algorithm, with the primary difference being the use of the quantum generator. We optimize the model using an alternating training strategy. Within a training epoch, the parameters of the generator $G$ are first fixed, and the discriminator $D$ is trained. 
Subsequently, the parameters of the discriminator $D$ are fixed, and the generator $G$ is trained.
These two steps are repeatedly alternated until a Nash equilibrium is reached. 
The complete training algorithm for BasicQGAN is detailed in Algorithm~\ref{alg:BasicQGAN algorithm}.

\begin{algorithm}[!h]
\caption{Pseudocode of training algorithm for BasicQGAN}
\label{alg:BasicQGAN algorithm}
\textbf{Input}: Gradient penalty coefficient $\lambda$, number of epochs $n_{\text{epoch}}$, batch size $m$, critic update every $n_c$ iterations, Adam momentum hyperparameters $b_1$, $b_2$, the learning rate hyperparameter of the generator $\eta_1$, the learning rate hyperparameter of the critic $\eta_2$, random number $\xi \sim U[0, 1]$, optimized rotation angle sampling distribution $Z'\sim P(Z')$. 
\begin{algorithmic}[1] %[1] enables line numbers
\STATE Initialize discriminator parameters $\omega$, generator parameters $\theta$.
% $\theta$ consists of parameters $\alpha$, $\beta$, $\delta$.
\STATE $batchnum:=0$
\FOR {$\text{epoch} = 1, \ldots, n_{\text{epoch}}$}
\FOR{$i = 1, \ldots, m$}
\STATE Sample real data $x \sim P_{\text{data}}(x)$, rotation angle $Z'\sim P(Z')$
\STATE $x' \stackrel{\text{measurement}}{\longleftarrow}\left| \psi \right\rangle \leftarrow G(\theta,Z')$
\STATE $\hat{x} \leftarrow \xi x + (1 - \xi) x^{\prime}$
\STATE $ {L_{D}}\leftarrow D({x^{\prime}}) - D(x) + \lambda \left( \left\| \nabla_{\hat{x}} D(\hat{x}) \right\|_2 - 1 \right)^2$ 
\STATE $\omega \leftarrow \text{Adam} \left( \frac{1}{m} \sum_{i=1}^{m} {L_{D}},\omega, \eta_2, b_1, b_2 \right)$ 
\STATE $batchnum:=batchnum+1$ 
\IF{$batchnum \bmod n_c == 0$} 
\STATE $x^{\prime} \stackrel{\text{measurement}}{\longleftarrow}\left| \psi\right\rangle \leftarrow G(\theta, Z')$ 
\STATE $ L_{G}\leftarrow D({x^{\prime}}) $
\STATE $\theta\leftarrow \text{Adam} \left( \frac{1}{m} \sum_{i=1}^{m} L_{G}, \theta, \eta_1, b_1, b_2 \right)$   \ENDIF
\ENDFOR    
\ENDFOR
\end{algorithmic}
\end{algorithm}

\newpage
\section{Dataset Descriptions}
\label{app_5}
Most existing baselines for QGAN image generation are benchmarked on small-scale grayscale datasets. Following this common practice, we construct four datasets from public repositories (MNIST, EMNIST, and HDS), as shown in Figure~\ref{crop_Appendix_Dataset_line} and summarized in Table~\ref{tab_dataset}.

\begin{figure}[H]
    \centering
   \includegraphics[width=\linewidth, trim=0cm 0cm 0cm 0cm, clip]{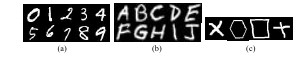}
    \caption{Representative samples from our used datasets: (a) MNIST, (b) Uppercase Letters, and (c) Geometric Shapes.}
    \label{crop_Appendix_Dataset_line}
\end{figure}

\begin{table}[H]
  \caption{A summary of the datasets.}
  \label{tab_dataset}
  \begin{center}
    \begin{small}
      \scshape
        \setlength{\tabcolsep}{1mm}
        \begin{tabular}{@{}cccc@{}}
          \toprule
          \makecell{Dataset} &
          \makecell{Image\\ Resolution} &
          \makecell{No. of\\ Training Instances} &
          \makecell{No. of\\ Classes} \\
          \midrule
          MNIST & $28\times 28$ & 1500 & 10 object categories \\
          MNIST & $16\times 16$ & 1500 & 10 object categories \\
          \makecell{Uppercase\\ Letters} & $16\times 16$ & 1500 & 10 object categories \\
          \makecell{Geometric\\ Shapes}  & $16\times 16$ & 400 & 4 object categories \\
          \bottomrule
        \end{tabular}
    \end{small}
  \end{center}
  \vskip -0.1in
\end{table}

\section{Additional Experimental Results and Training Details}
\label{app_6}
This appendix provides supplementary experimental results and further details regarding the training processes mentioned in the main paper. Figure~\ref{crop_16size0to9}, Figure~\ref{crop_28size0to9}, Figure~\ref{crop_16sizeGeometric2}, and Figure~\ref{crop_16sizeUppercase} showcase more generated image samples from BasicQGAN, PQWGAN, and WGAN-GP across various datasets.
In addition, Figure~\ref{crop_FID_loss}  reports the FID curves over training epochs on MNIST at two resolutions.

\begin{figure}[H]
    \centering
   \includegraphics[width=0.8\linewidth, trim=0cm 0cm 0cm 0cm, clip]{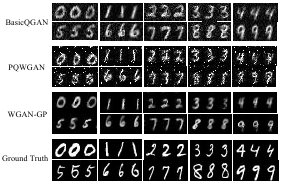}
    \caption{The generated samples for all digit classes (0-9) produced by various models on the 16$\times$16 MNIST dataset.}
    \label{crop_16size0to9}
\end{figure}

\begin{figure}[H]
    \centering
   \includegraphics[width=0.8\linewidth, trim=0cm 0cm 0cm 0cm, clip]{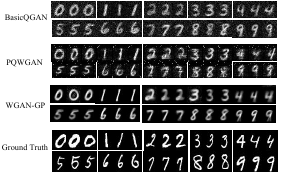}
    % \caption{Representative sample images from the datasets used in our experiments: (a) MNIST digits, (b) Uppercase Letters, and (c) Geometric Shapes.}
    \caption{The generated samples for all digit classes (0-9) produced by various models on the 28$\times$28 MNIST dataset.}
    \label{crop_28size0to9}
\end{figure}

\begin{figure}[H]
    \centering
   \includegraphics[width=0.9\linewidth, trim=0cm 0cm 0cm 0cm, clip]{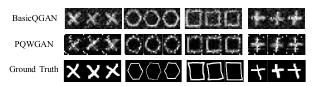}
    \caption{The generated samples for all shape classes (angleCross,hexagon,  square and straightCross) produced by various models on the 16$\times$16 Geometric Shapes dataset.}
    \label{crop_16sizeGeometric2}
\end{figure}

\begin{figure}[H]
    \centering
   \includegraphics[width=0.9\linewidth, trim=0cm 0cm 0cm 0cm, clip]{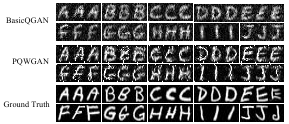}
    \caption{The generated samples (A-J) for all letter classes produced by various models on the 16$\times$16 Uppercase Letters dataset.}
    \label{crop_16sizeUppercase}
\end{figure}

\begin{figure}[H]
    \centering
   \includegraphics[width=0.8\linewidth, trim=0cm 0cm 0cm 0cm, clip]{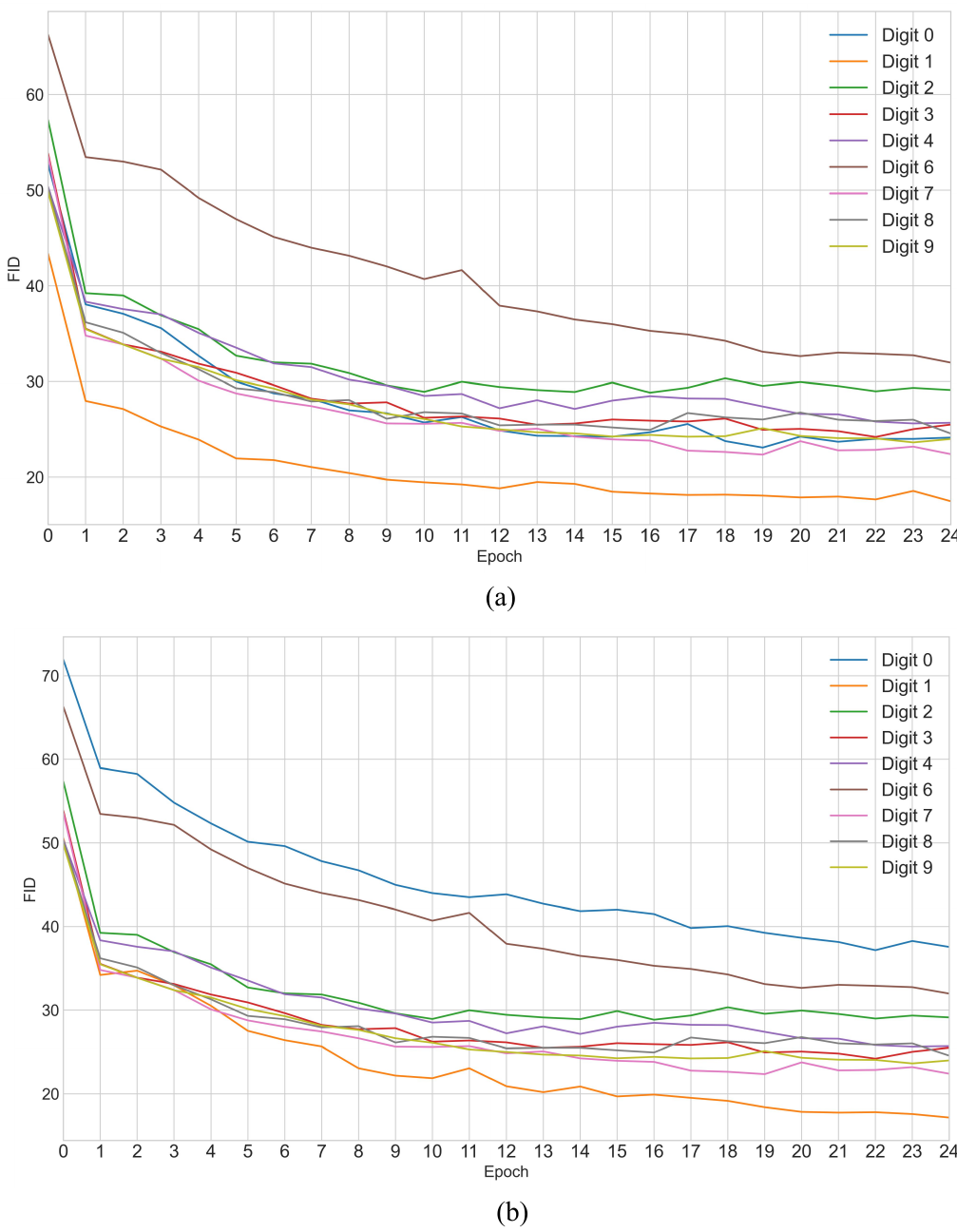}
    % \caption{Representative sample images from the datasets used in our experiments: (a) MNIST digits, (b) Uppercase Letters, and (c) Geometric Shapes.}
    \caption{FID curves over training epochs on MNIST at two resolutions: (a) $16\times16$ and (b) $28\times28$. Each line reports the per-digit FID, showing a consistent decrease and gradual convergence as training proceeds.}
    \label{crop_FID_loss}
\end{figure}
\newpage

\section{Analysis of quantum circuit layers}
\label{app_7}
In this section, we analyze the effect of quantum-generator circuit depth on generation performance. As shown in Figure~\ref{differentlayers}, circuits with 40 layers generate clearer and more diverse images than circuits with 10 or 20 layers. This improvement is consistent with the increased expressive power provided by deeper parameterized circuits. However, further increasing the depth to 60 or 80 layers reduces sample diversity, likely because deeper variational circuits are more susceptible to barren plateaus and harder to optimize effectively. Table~\ref{tab:fid_layers} reports the corresponding FID scores. The best score is achieved with 40 layers, after which FID increases.

\begin{figure}[H]
    \centering
   \includegraphics[width=1\textwidth, trim=0cm 0cm 0cm 0cm, clip]{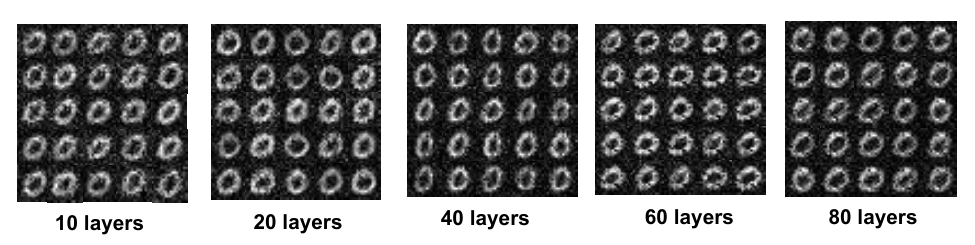}
    \caption{Qualitative comparison of BasicQGAN generated samples on the $16\times16$ MNIST setting with different generator depths (10, 20, 40, 60, and 80 layers). The 40-layer configuration is used as our default setting.}
    \label{differentlayers}
\end{figure}

\begin{table}[H]
  \caption{Comparison of BasicQGAN performance across different generator depths.}
  \label{tab:fid_layers}
  \begin{center}
    \begin{small}
      \scshape
        \begin{tabular}{lccccc}
          \toprule
          Layer & 10 & 20 & 40 (Ours) & 60 & 80 \\
          \midrule
          FID   & 25.12 & 20.71 & \textbf{20.38} & 21.74 & 26.24 \\
          \bottomrule
        \end{tabular}
    \end{small}
  \end{center}
  \vskip -0.1in
\end{table}

\section{Analysis of different learning rate}
\label{app_8}
This section analyzes how the learning rate in the 
Quantum Fidelity Landscape (QFL) optimization algorithm
affects the generation performance of BasicQGAN. To quantify how well the optimized initial-state QFL matches the true-sample QFL, we measure the discrepancy between the two distributions using the 1D Wasserstein distance. Following the sum form (i.e., omitting the $1/n$ scaling), for two sorted sets of discrete values $A=\{a_1\le\cdots\le a_n\}$ and $B=\{b_1\le\cdots\le b_n\}$, we compute
\begin{equation}
W_1(A,B)=\sum_{i=1}^{n}\lvert a_{(i)}-b_{(i)}\rvert .
\end{equation}

Figure~\ref{lr_learningrate}(a)--(d) visualizes the alignment between the QFL of the optimized initial-state ensemble and that of the true-sample ensemble under different learning rates, and reports the corresponding Wasserstein distances for each setting. A smaller Wasserstein distance indicates a closer match to the true-sample QFL. In addition, Figure~\ref{lr_learningrate}(e) shows the average-fidelity training curves during the initial-state sampling optimization.

Overall, the optimized initial-state QFLs obtained with different learning rates are all close to the true-sample QFL. The smallest Wasserstein distances are achieved with learning rates of 0.5 and 0.3 (0.0890 and 0.0882, respectively). Moreover, the average-fidelity curve is smoother under a learning rate of 0.3, indicating more stable convergence. Considering both matching quality and training stability, we choose a learning rate of 0.3.

\begin{figure}[H]
    \centering
   \includegraphics[width=0.9\textwidth, trim=0cm 0cm 0cm 0cm, clip]{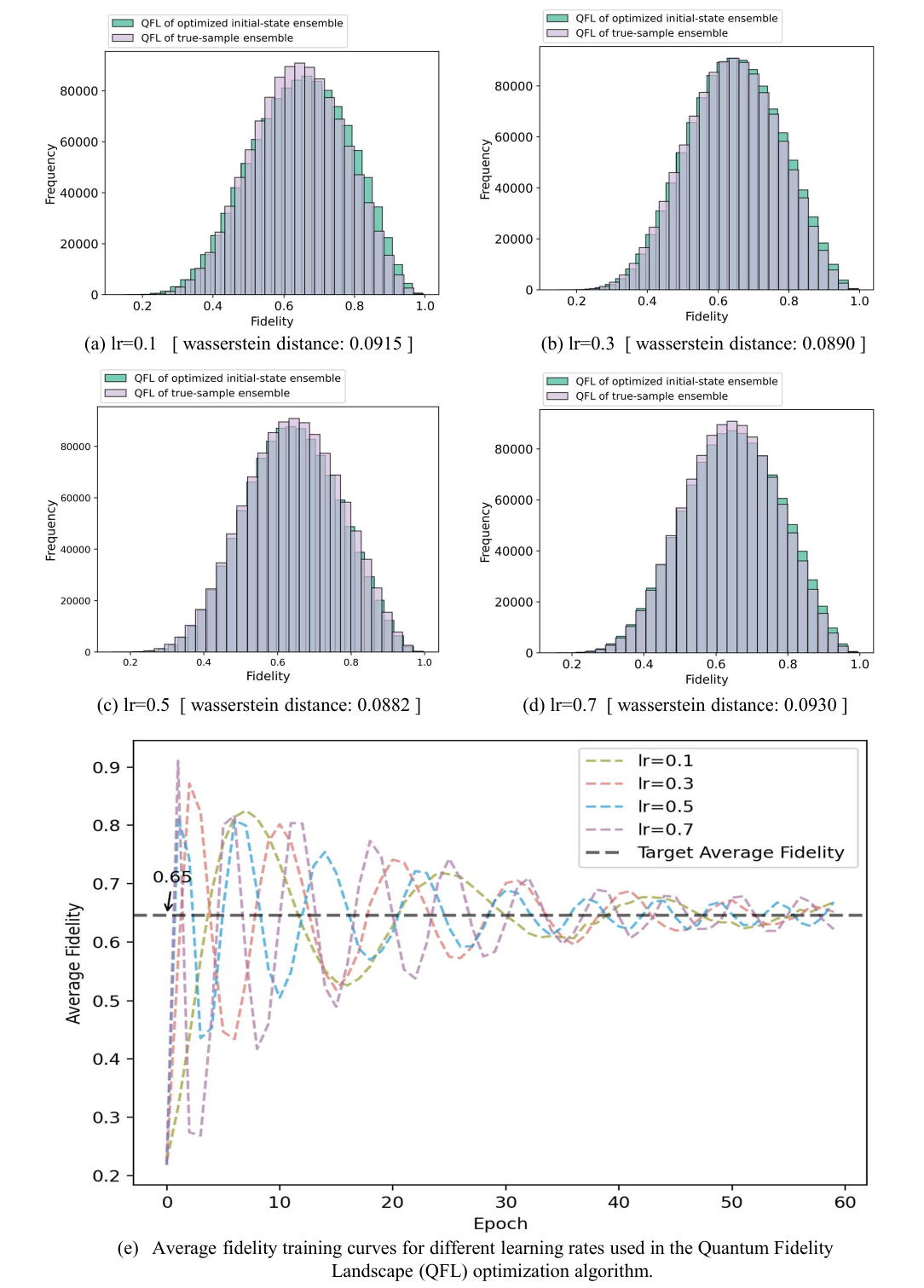}
    \caption{Effect of the different learning rate in  Quantum Fidelity Landscape (QFL) optimization algorithm. (a)--(d) histograms of QFL values for the optimized initial-state ensemble and the true-sample ensemble under different learning rates; the corresponding 1D Wasserstein distances are reported in each subplot. (e) Average-fidelity training curves for different learning rates during the optimization, where the dashed line denotes the target average fidelity.}  
    \label{lr_learningrate}
\end{figure}

\clearpage

\bibliographystyle{elsarticle-harv}
\bibliography{references}

\end{document}